\documentclass[pra,aps,showpacs,onecolumn,twoside,superscriptaddress,9pt]{revtex4}

\usepackage{mathtools}
\usepackage{amsmath,amsfonts,amssymb,caption,color,epsfig,graphics,graphicx,hyperref,latexsym,mathrsfs,revsymb,theorem,url,verbatim,epstopdf,float}
\usepackage{algorithm}
\usepackage{algpseudocode}
\usepackage{multirow}
\usepackage{tikz}
\usepackage[section]{placeins} 
\usepackage[level]{datetime}
\usetikzlibrary{shapes, arrows}

\graphicspath{{images/}}

\hypersetup{colorlinks,linkcolor={blue},citecolor={red},urlcolor={blue}}

\newtheorem{definition}{Definition}
\newtheorem{proposition}[definition]{Proposition}
\newtheorem{lemma}[definition]{Lemma}

\newtheorem{theorem}[definition]{Theorem}
\newtheorem{corollary}[definition]{Corollary}
\newtheorem{conjecture}[definition]{Conjecture}

\newtheorem{remark}[definition]{Remark}
\newtheorem{example}[definition]{Example}
\newtheorem{question}[definition]{Question}

\def\bcj{\begin{conjecture}}
\def\ecj{\end{conjecture}}
\def\bcr{\begin{corollary}}
\def\ecr{\end{corollary}}
\def\bd{\begin{definition}}
\def\ed{\end{definition}}
\def\bea{\begin{eqnarray}}
\def\eea{\end{eqnarray}}
\def\bem{\begin{enumerate}}
\def\eem{\end{enumerate}}
\def\bex{\begin{example}}
\def\eex{\end{example}}
\def\bim{\begin{itemize}}
\def\eim{\end{itemize}}
\def\bl{\begin{lemma}}
\def\el{\end{lemma}}
\def\bma{\begin{bmatrix}}
\def\ema{\end{bmatrix}}
\def\bpf{\begin{proof}}
\def\epf{\end{proof}}
\def\bpp{\begin{proposition}}
\def\epp{\end{proposition}}
\def\bqu{\begin{question}}
\def\equ{\end{question}}
\def\br{\begin{remark}}
\def\er{\end{remark}}
\def\bt{\begin{theorem}}
\def\et{\end{theorem}}

\def\squareforqed{\hbox{\rlap{$\sqcap$}$\sqcup$}}
\def\qed{\ifmmode\squareforqed\else{\unskip\nobreak\hfil
\penalty50\hskip1em\null\nobreak\hfil\squareforqed
\parfillskip=0pt\finalhyphendemerits=0\endgraf}\fi}
\def\endenv{\ifmmode\;\else{\unskip\nobreak\hfil
\penalty50\hskip1em\null\nobreak\hfil\;
\parfillskip=0pt\finalhyphendemerits=0\endgraf}\fi}
\newenvironment{proof}{\noindent \textbf{{Proof.~} }}{\qed}
\def\Dbar{\leavevmode\lower.6ex\hbox to 0pt
{\hskip-.23ex\accent"16\hss}D}
\makeatletter
\def\url@leostyle{%
  \@ifundefined{selectfont}{\def\UrlFont{\sf}}{\def\UrlFont{\small\ttfamily}}}
\makeatother
\def\a{\alpha}

\def\g{\gamma}
\def\d{\delta}
\def\e{\epsilon}
\def\ve{\varepsilon}

\def\t{\theta}

\def\k{\kappa}
\def\l{\lambda}
\def\m{\mu}

\def\p{\pi}
\def\r{\rho}
\def\s{\sigma}

\def\ps{\psi}

\def\D{\Delta}
\def\T{\Theta}

\def\Ph{\Phi}

\newcommand{\nc}{\newcommand}

 \nc{\bbA}{\mathbb{A}} \nc{\bbB}{\mathbb{B}} \nc{\bbC}{\mathbb{C}}
 \nc{\bbD}{\mathbb{D}} \nc{\bbE}{\mathbb{E}} \nc{\bbF}{\mathbb{F}}
 \nc{\bbG}{\mathbb{G}} \nc{\bbH}{\mathbb{H}} \nc{\bbI}{\mathbb{I}}
 \nc{\bbJ}{\mathbb{J}} \nc{\bbK}{\mathbb{K}} \nc{\bbL}{\mathbb{L}}
 \nc{\bbM}{\mathbb{M}} \nc{\bbN}{\mathbb{N}} \nc{\bbO}{\mathbb{O}}
 \nc{\bbP}{\mathbb{P}} \nc{\bbQ}{\mathbb{Q}} \nc{\bbR}{\mathbb{R}}
 \nc{\bbS}{\mathbb{S}} \nc{\bbT}{\mathbb{T}} \nc{\bbU}{\mathbb{U}}
 \nc{\bbV}{\mathbb{V}} \nc{\bbW}{\mathbb{W}} \nc{\bbX}{\mathbb{X}}
 \nc{\bbZ}{\mathbb{Z}}

 \nc{\bA}{{\bf A}} \nc{\bB}{{\bf B}} \nc{\bC}{{\bf C}}
 \nc{\bD}{{\bf D}} \nc{\bE}{{\bf E}} \nc{\bF}{{\bf F}}
 \nc{\bG}{{\bf G}} \nc{\bH}{{\bf H}} \nc{\bI}{{\bf I}}
 \nc{\bJ}{{\bf J}} \nc{\bK}{{\bf K}} \nc{\bL}{{\bf L}}
 \nc{\bM}{{\bf M}} \nc{\bN}{{\bf N}} \nc{\bO}{{\bf O}}
 \nc{\bP}{{\bf P}} \nc{\bQ}{{\bf Q}} \nc{\bR}{{\bf R}}
 \nc{\bS}{{\bf S}} \nc{\bT}{{\bf T}} \nc{\bU}{{\bf U}}
 \nc{\bV}{{\bf V}} \nc{\bW}{{\bf W}} \nc{\bX}{{\bf X}}
 \nc{\bZ}{{\bf Z}}

\nc{\cA}{{\cal A}} \nc{\cB}{{\cal B}} \nc{\cC}{{\cal C}}
\nc{\cD}{{\cal D}} \nc{\cE}{{\cal E}} \nc{\cF}{{\cal F}}
\nc{\cG}{{\cal G}} \nc{\cH}{{\cal H}} \nc{\cI}{{\cal I}}
\nc{\cJ}{{\cal J}} \nc{\cK}{{\cal K}} \nc{\cL}{{\cal L}}
\nc{\cM}{{\cal M}} \nc{\cN}{{\cal N}} \nc{\cO}{{\cal O}}
\nc{\cP}{{\cal P}} \nc{\cQ}{{\cal Q}} \nc{\cR}{{\cal R}}
\nc{\cS}{{\cal S}} \nc{\cT}{{\cal T}} \nc{\cU}{{\cal U}}
\nc{\cV}{{\cal V}} \nc{\cW}{{\cal W}} \nc{\cX}{{\cal X}}
\nc{\cZ}{{\cal Z}}

\nc{\hA}{{\hat{A}}} \nc{\hB}{{\hat{B}}} \nc{\hC}{{\hat{C}}}
\nc{\hD}{{\hat{D}}} \nc{\hE}{{\hat{E}}} \nc{\hF}{{\hat{F}}}
\nc{\hG}{{\hat{G}}} \nc{\hH}{{\hat{H}}} \nc{\hI}{{\hat{I}}}
\nc{\hJ}{{\hat{J}}} \nc{\hK}{{\hat{K}}} \nc{\hL}{{\hat{L}}}
\nc{\hM}{{\hat{M}}} \nc{\hN}{{\hat{N}}} \nc{\hO}{{\hat{O}}}
\nc{\hP}{{\hat{P}}} \nc{\hR}{{\hat{R}}} \nc{\hS}{{\hat{S}}}
\nc{\hT}{{\hat{T}}} \nc{\hU}{{\hat{U}}} \nc{\hV}{{\hat{V}}}
\nc{\hW}{{\hat{W}}} \nc{\hX}{{\hat{X}}} \nc{\hZ}{{\hat{Z}}}

\nc{\hn}{{\hat{n}}}

\def\dim{\mathop{\rm Dim}}

\def\max{\mathop{\rm max}}
\def\min{\mathop{\rm min}}

\def\tr{\mathop{\rm Tr}}

\def\bigox{\bigotimes}
\def\dg{\dagger}

\def\ox{\otimes}

\def\ra{\rightarrow}

\newcommand{\bra}[1]{\langle#1|}
\newcommand{\ket}[1]{|#1\rangle}

\newcommand{\norm}[1]{\lVert#1\rVert}

\newcommand{\by}{\mathbf{y}}
\newcommand{\bh}{\mathbf{h}}
\newcommand{\bb}{\mathbf{b}}
\newcommand{\ba}{\mathbf{a}}
\newcommand{\bm}{\mathbf{m}}

\newcommand{\bs}{\boldsymbol{\sigma}}

\def\Dbar{\leavevmode\lower.6ex\hbox to 0pt
{\hskip-.23ex\accent"16\hss}D}
\usepackage{color}
\begin{document}
\title{Physics-Informed Neural Networks with Adaptive Constraints for Multi-Qubit Quantum Tomography}

\newdateformat{ukdate}{\ordinaldate{\THEDAY} \monthname[\THEMONTH] \THEYEAR}
\date{\ukdate\today}


\author{Changchun Feng}
\affiliation{Hangzhou International Innovation Institute, Beihang University, Hangzhou, China}

\author{Laifa Tao}\email[]{taolaifa@buaa.edu.cn (corresponding author)}
\affiliation{Hangzhou International Innovation Institute, Beihang University, Hangzhou, China}
\affiliation{School of Reliability and Systems Engineering, Beihang University, Beijing, China}
\affiliation{Institute of Reliability Engineering, Beihang University, Beijing, China}
\affiliation{Science $\&$ Technology on Reliability $\&$ Environmental Engineering Laboratory, Beijing, China}

\author{Lin Chen}\email[]{linchen@buaa.edu.cn (corresponding author)}
\affiliation{LMIB(Beihang University), Ministry of Education, and School of Mathematical Sciences, 
Beihang University, Beijing 100191, China}

\begin{abstract}
Quantum state tomography (QST) faces exponential measurement requirements and noise sensitivity in multi-qubit systems, creating fundamental bottlenecks in practical quantum technologies. This paper presents a physics-informed neural network (PINN) framework integrating quantum mechanical constraints through adaptive constraint weighting, a residual-and-attention-enhanced architecture, and differentiable Cholesky parameterization ensuring physical validity. Comprehensive evaluations across 2--5 qubit systems and arbitrary-dimensional quantum states demonstrate PINN consistently outperforms Traditional Neural Network, achieving the highest fidelity at every tested dimension. PINN exhibits significantly higher fidelity compared to baseline methods, with particularly dramatic improvements in moderately high-dimensional systems. The approach demonstrates superior noise robustness, showing substantially slower performance degradation under increasing noise levels compared to traditional neural networks. Across diverse dimensional configurations, PINN achieves consistent improvements over Traditional Neural Network, demonstrating superior dimensional robustness. Theoretical analysis establishes that physics constraints reduce Rademacher complexity and mitigate the curse of dimensionality through constraint-induced dimension reduction and sample complexity reduction mechanisms that remain effective regardless of qubit number. While experimental validation is conducted up to 5-qubit systems due to computational constraints, the theoretical framework---including convergence guarantees, generalization bounds, and scalability theorems---provides rigorous justification for expecting PINN's advantages to persist and strengthen in larger systems (6 qubits and beyond), where the relative benefit of constraint-induced dimension reduction grows with system size. The practical significance extends to quantum error correction and gate calibration, where PINN reduces measurement requirements from $O(4^n)$ to $O(2^n)$ while maintaining high fidelity, enabling faster error correction cycles and accelerated calibration workflows critical for scalable quantum computing.

\end{abstract}

\maketitle

\section{Introduction}

Quantum state tomography (QST) serves as a fundamental tool in quantum 
information science, enabling validation of quantum processors, calibration 
of quantum gates, and monitoring of decoherence in distributed quantum 
links 
\cite{PRXQuantum.2.010318,PhysRevLett.84.5748,2010CMP,PhysRevLett.105.150401,10417060,
PhysRevLett.126.100402,PhysRevLett.109.120403,PhysRevLett.91.090406}. QST reconstructs 
the complete density matrix of an $n$-qubit system from informationally 
complete measurements, thereby enabling downstream tasks such as entanglement certification, 
resource quantification, and feedback 
control \cite{PhysRevA.66.012303,PhysRevA.101.052316}. However, the exponential growth of measurement 
requirements with qubit number ($O(4^n)$) creates fundamental bottlenecks that severely limit the feasibility 
of QST for multi-qubit systems, where measurement overhead becomes prohibitive. Furthermore, realistic noise 
conditions—including amplitude damping, phase decoherence, correlated $Z$-dephasing, and crosstalk—rapidly degrade 
reconstruction accuracy \cite{PhysRevA.109.022413,PhysRevLett.122.190401}, mirroring challenges encountered in multipartite entanglement certification that demand data-efficient and structure-aware learning strategies \cite{PhysRevResearch.6.023250}.

The critical importance of efficient QST extends beyond fundamental characterization 
to practical quantum technologies that are essential for scalable quantum computing. 
In quantum error correction protocols, real-time state reconstruction is required for 
syndrome extraction and error detection, where measurement efficiency directly impacts 
error correction cycle times and determines the feasibility of fault-tolerant quantum 
computation \cite{PhysRevA.107.062409,CSL2024}. Current surface code implementations 
require frequent state tomography to monitor logical qubit states, but traditional QST 
methods impose prohibitive measurement overhead ($O(4^n)$) that limits 
code performance and prevents real-time error correction \cite{PRXQuantum.2.010318}. 
Similarly, quantum gate calibration relies on iterative state reconstruction to 
optimize gate parameters, where the exponential measurement scaling of classical 
methods creates bottlenecks in calibration workflows, significantly increasing the 
time and cost required for quantum hardware optimization \cite{fcc_ctp2024,JBP2017}. 
The development of measurement-efficient QST methods is therefore not merely an 
academic exercise but a prerequisite for scalable quantum computing, directly addressing 
fundamental bottlenecks that have limited the development of large-scale quantum systems.

To address these challenges, existing approaches to quantum state tomography can be 
systematically categorized into three classes: classical optimization-based methods, 
machine-learning-based estimators, and hybrid approaches \cite{PhysRevLett.127.260401,Wu2023overlappedgrouping}.
 Traditional classical methods include least squares reconstruction using linear 
 inversion \cite{PhysRevLett.105.150401}, maximum-likelihood estimation (MLE) with 
 iterative optimization algorithms such as R$\rho$R \cite{PhysRevLett.105.150401}, compressed 
 sensing \cite{PhysRevLett.109.120403}, and Bayesian reconstruction. These 
 methods reduce measurement requirements by assuming sparsity or low rank. However, 
 they face fundamental limitations: they require explicit optimization over the full 
 parameter space, leading to computational complexity that scales exponentially with 
 qubit number \cite{Zhang2022Variational,PhysRevResearch.2.043158}. Once the qubit count exceeds four, 
 classical methods face an unavoidable trade-off between accuracy, runtime, and physical validity \cite{PRXQuantum.2.020348,titchener2018scalable}. 
 The computational burden becomes prohibitive for larger systems, and these methods struggle to maintain physical 
 validity while achieving high reconstruction fidelity \cite{MENG2023106661}.

To overcome these limitations, recent advances have explored neural network approaches that amortize inference over large datasets 
and can learn nonlinear priors \cite{2018TGMG}. Neural networks have demonstrated success in related quantum tasks 
including quantum state classification \cite{PERALGARCIA2024100619,schuld2019supervised}, 
gate optimization \cite{fcc_ctp2024,JBP2017}, and error 
correction \cite{CSL2024,KBD2020,PhysRevA.107.062409}. Early QST studies 
report higher fidelity with fewer measurement settings compared to classical
 methods \cite{PhysRevA.106.012409,PhysRevLett.127.260401}. However, purely 
 data-driven models face three critical limitations: (1) they require many 
 labeled samples for training, limiting practical applicability; (2) they 
 generalize poorly to distribution shifts, failing when test conditions differ 
 from training conditions; (3) they may output nonphysical density matrices because 
 fundamental quantum mechanical constraints—Hermiticity, trace normalization, and 
 positivity—are not explicitly enforced \cite{PhysRevA.106.012409,PhysRevA.108.022427,10564586,PERALGARCIA2024100619,10.1093/nsr/nwaf269}. The mismatch between training and deployment conditions is particularly severe on real quantum hardware, where drift and non-Markovian noise constantly change measurement statistics, rendering purely data-driven approaches unreliable.

To address these fundamental limitations, physics-informed neural networks (PINNs) represent 
an emerging paradigm that combines data-driven learning with 
domain knowledge \cite{RAISSI2019686,YUAN2022111260}. Similar to 
semisupervised entanglement detection that mixes witnesses with 
data \cite{2024JGC,YW2024}, PINNs penalize violations of governing 
equations and maintain robustness when data are scarce or noisy. However, 
their application to QST remains limited. 
Our physics-informed neural network approach addresses these practical 
challenges by achieving high-fidelity reconstruction with significantly 
reduced measurement requirements. The adaptive constraint weighting mechanism 
enables robust performance under realistic noise conditions encountered in quantum 
hardware, automatically adjusting constraint strength based on noise severity to 
balance physical validity with data fitting. The Cholesky parameterization ensures 
physical validity critical for downstream error correction and calibration tasks, 
guaranteeing that all reconstructed density matrices satisfy fundamental quantum 
mechanical constraints. These capabilities position PINN as a practical tool for 
real-world quantum information processing applications, directly enabling faster 
error correction cycles in quantum error correction protocols and accelerated 
calibration workflows in quantum gate optimization. The comprehensive experimental 
validation presented in this work demonstrates that PINN maintains consistent advantages 
across diverse dimensional regimes, from standard qubit systems to arbitrary-dimensional 
quantum states, establishing its broad applicability for practical quantum technologies. 
Although experimental evaluation is conducted up to 5-qubit systems due to computational resource constraints, the theoretical framework provides rigorous justification for expecting PINN's advantages to persist and even strengthen in larger systems (6 qubits and beyond). The established convergence guarantees, generalization bounds, and scalability theorems demonstrate that physics constraints mitigate the curse of dimensionality through constraint space dimension reduction and sample complexity reduction mechanisms that remain effective regardless of qubit number. These theoretical guarantees ensure that the benefits of physics-informed learning become increasingly valuable as system size increases, positioning PINN as a scalable approach for quantum state tomography that can be confidently deployed in practical quantum computing applications.

This paper presents a PINN framework for multi-qubit QST that bridges purely 
data-driven neural estimators and traditional physics-based solvers. The contributions 
are threefold. (1) Methodologically, an adaptive physics-constraint weighting strategy, 
a residual-and-attention-enhanced multi-task PINN architecture, and a differentiable 
Cholesky parameterization that enforces positive semidefiniteness are developed. 
(2) Experimentally, comprehensive evaluations across 2--5 qubit systems and arbitrary-dimensional quantum states (up to $10 \times 10$) comparing PINN against baseline methods demonstrate consistent fidelity improvements across all tested dimensions. The most dramatic advantage of 30.3\% occurs in 4-qubit systems compared to Traditional Neural Network (0.8872 vs 0.6810). Across nine dimensional configurations, PINN achieves the highest fidelity at every dimension, with an average improvement of +2.40\% over Traditional Neural Network. PINN exhibits superior dimensional robustness (41.5\% performance variation) and demonstrates a 2.6× slower performance degradation rate under noise compared to baseline methods. (3) Theoretically, convergence and generalization bounds are established that explain how physics priors reduce Rademacher complexity and mitigate the curse of dimensionality, providing rigorous justification for the observed empirical advantages. While experimental validation is conducted up to 5-qubit systems due to computational resource constraints, the theoretical framework establishes that PINN's advantages extend to higher-dimensional systems (6 qubits and beyond) through constraint space dimension reduction and sample complexity reduction mechanisms that remain effective regardless of qubit number.

The remainder of this paper is organized as follows: Section \ref{sec:pri} reviews 
the theoretical foundations of quantum state tomography and physics-informed 
neural networks, establishing the mathematical framework and notation. 
Section \ref{sec:method} details the PINN architecture, including adaptive constraint weighting mechanisms, residual-and-attention-enhanced design, and Cholesky parameterization, along with training procedures. Section \ref{sec:experiments} presents comprehensive experimental validation through three complementary studies: (1) high-fidelity quantum link entanglement health monitoring that validates practical applicability under realistic noise conditions, (2) systematic scalability evaluation across standard qubit systems (2--5 qubits) and arbitrary-dimensional quantum states (up to $10 \times 10$), and (3) ablation studies that isolate and quantify the contribution of each architectural component. These experiments systematically compare PINN against Traditional Neural Network, Least Squares, and Maximum Likelihood Estimation, revealing consistent advantages across diverse system configurations and dimensional regimes. Section \ref{sec:theory} provides rigorous theoretical analysis establishing convergence guarantees, generalization bounds, and scalability properties that explain the observed empirical advantages, particularly the dimensional threshold effects and robustness improvements. Section \ref{sec:results} discusses the implications of the results, limitations, and future research directions. Section \ref{sec:applications} details practical applications in quantum error correction and gate calibration, demonstrating the real-world impact of the approach through quantified measurement reduction and speed improvements. Finally, Section \ref{sec:conclusion} summarizes the key contributions and their significance for quantum state tomography and quantum information processing, highlighting how the comprehensive experimental validation establishes PINN as a practical tool for scalable quantum computing applications.

\section{Preliminary}
\label{sec:pri}

Quantum state tomography (QST) serves as the fundamental bridge between quantum measurements and the complete characterization of quantum systems. To establish a foundation for the physics-informed neural network approach presented herein, the essential mathematical framework of quantum states and the computational challenges inherent in their reconstruction are reviewed.

\subsection{Quantum State Tomography}

For an $n$-qubit quantum system, the density matrix $\r \in \bbC^{D \times D}$ with 
$D = 2^n$ provides a complete description of the quantum state. A valid density
 matrix must satisfy three fundamental constraints: Hermiticity ($\r = \r^\dg$), 
 positivity ($\r \succeq 0$), and trace normalization ($\tr(\r) = 1$). 
 These constraints ensure that all eigenvalues are real and non-negative. 
 They also ensure that the matrix represents a valid quantum state 
 with proper probability interpretation. For a pure state $\ket{\ps}$, the density matrix takes the form $\r = \ket{\ps}\bra{\ps}$. Mixed states are represented as convex combinations $\r = \sum_i p_i \ket{\ps_i}\bra{\ps_i}$ with $\sum_i p_i = 1$ and $p_i \geq 0$.

The measurement process in quantum systems is described by positive operator-valued measures (POVMs). In multi-qubit systems, Pauli measurement bases are typically employed. The expectation value of a Pauli operator $\hP$ is computed as $\langle \hP \rangle = \tr(\r \hP)$. Given measurement frequencies $\{f_i\}$ for different Pauli bases, the fundamental challenge of QST is to reconstruct the density matrix $\r$ from these incomplete and often noisy measurements.

This reconstruction task faces a critical computational barrier known as the curse of dimensionality. For an $n$-qubit system, the number of independent parameters in $\r$ scales as $O(4^n)$. This requires traditional QST methods to perform $O(4^n)$ measurements. This exponential scaling renders the problem computationally intractable for large $n$. This motivates the need for more efficient reconstruction approaches that can leverage prior knowledge and structural constraints.

To address these computational challenges, physics-informed neural networks provide a promising framework that incorporates domain knowledge directly into the learning process.

\subsection{Physics-Informed Neural Networks}

Physics-informed neural networks (PINNs) address the limitations of purely data-driven approaches by embedding physical laws directly into the learning process. The core principle involves augmenting the standard data fitting loss with physics constraint terms:
\begin{equation}
L(\t) = L_{\text{data}}(\t) + \l L_{\text{physics}}(\t),
\end{equation}
where $L_{\text{data}}$ captures the discrepancy between predicted and observed measurements. $L_{\text{physics}}$ enforces adherence to physical laws. This methodology has demonstrated remarkable success in solving partial differential equations. It leverages domain knowledge to reduce data requirements and improve generalization performance.

The application of PINNs to quantum state tomography represents a natural extension of this paradigm. The physical constraints of quantum mechanics serve as the governing equations that guide the learning process. By incorporating quantum mechanical principles directly into the loss function, PINNs can maintain physical validity even when training data is limited or contaminated with noise.

To ensure that the reconstructed quantum states satisfy fundamental physical constraints, an appropriate parameterization strategy is essential. The Cholesky decomposition provides an elegant solution that naturally enforces physical validity.

\subsection{Cholesky Parameterization}

To ensure that reconstructed density matrices satisfy the fundamental quantum mechanical constraints, Cholesky decomposition is employed as a parameterization strategy. For any positive semidefinite matrix $\r$, there exists a lower triangular matrix $L$ such that $\r = L L^\dg$. This parameterization naturally guarantees the positivity constraint, as the product of a matrix with its conjugate transpose is always positive semidefinite. The remaining constraints of Hermiticity and trace normalization are enforced through additional structural constraints or normalization steps during the reconstruction process.

The Cholesky parameterization ensures physical validity and provides an efficient representation by reducing the effective dimensionality of the optimization problem. Working with triangular matrix parameters rather than the full density matrix maintains the essential quantum mechanical properties while enabling more effective optimization within the neural network framework.

Having established the theoretical foundations, we now present the comprehensive physics-informed neural network framework for multi-qubit quantum state tomography. This framework systematically addresses the fundamental challenges identified in the preliminary section through an integrated architecture that combines adaptive constraint mechanisms with advanced neural network design principles.

\section{Method}
\label{sec:method}

Given measurement data $\bm \in \bbR^d$ from an $n$-qubit system, the primary objective is to reconstruct the density matrix $\r \in \bbC^{D \times D}$ with $D = 2^n$. The neural network $f_\t: \bbR^d \ra \bbR^p$ serves as the core mapping function, transforming raw measurements into Cholesky parameters, where $p = D^2$ represents the parameter space dimension. Subsequently, the density matrix is reconstructed through the mapping $\Ph: \bbR^p \ra \cS_+$, where $\cS_+$ denotes the set of positive semidefinite matrices. This formulation ensures that reconstructed quantum states inherently satisfy the fundamental physical constraints required for valid density matrices.

The physics-informed neural network architecture, illustrated in Figure \ref{fig:architecture}, incorporates several key innovations that enhance both representation capacity and training stability. The network begins with an input layer that processes measurement data $\bm \in \bbR^d$, followed by multiple fully connected hidden layers enhanced with residual connections that facilitate gradient flow in deep networks. An attention mechanism is integrated to enable adaptive feature weighting, allowing the model to focus on the most informative measurements. The architecture culminates in a multi-task output layer that generates both the primary Cholesky parameters $\by \in \bbR^p$ for density matrix reconstruction and an auxiliary noise severity score $s(\bm)$. These outputs feed into the physics constraint module, where the noise score dynamically modulates the constraint weights, ensuring robust reconstruction across varying noise regimes.

\begin{figure}[ht]
\centering
\includegraphics[width=1.0\linewidth]{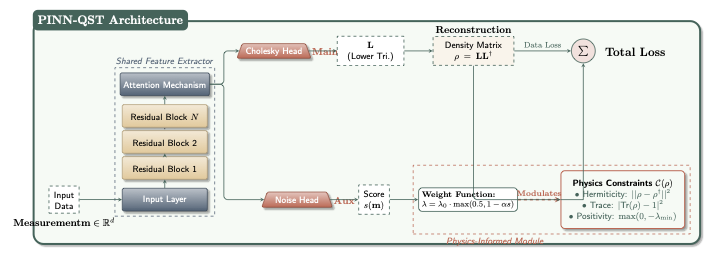}
\caption{Architecture of the physics-informed neural network for quantum state tomography. The network features residual connections in hidden layers, an attention mechanism for adaptive feature weighting, and multi-task learning with both main and auxiliary outputs. Physics constraints and Cholesky decomposition ensure the generation of physically valid density matrices.}
\label{fig:architecture}
\end{figure}

The residual connections are implemented at each layer according to the formulation:
\begin{equation}
\bh^{(l+1)} = \sigma(W^{(l)} \bh^{(l)} + \bb^{(l)}) + \text{Proj}(\bh^{(l)}),
\end{equation}
where $\text{Proj}$ represents a projection layer when dimensions differ between consecutive layers, and $\sigma$ denotes the activation function. This design mitigates the vanishing gradient problem and enables more effective training of deep networks.

The attention mechanism enhances feature selectivity through adaptive weighting:
\begin{equation}
\ba = \text{Sigmoid}(W_2 \text{ReLU}(W_1 \bh)),
\end{equation}
\begin{equation}
\bh_{\text{att}} = \bh \odot \ba,
\end{equation}
where $\odot$ denotes element-wise multiplication. This mechanism enables the model to dynamically focus on important measurement features while suppressing irrelevant information, which is particularly crucial in noisy quantum measurement scenarios.

A key architectural innovation is the multi-task learning framework, where the network simultaneously predicts both the main task output Cholesky parameters $\by$ for density matrix reconstruction and an auxiliary task output noise severity $s(\bm)$ for adaptive weight adjustment. This multi-task design creates a synergistic relationship in which the noise severity prediction directly informs the dynamic constraint weighting mechanism, enabling more robust performance across varying noise conditions.

Traditional approaches to physics-informed learning typically employ fixed constraint weights, which face a fundamental trade-off: high weights may cause underfitting in noisy scenarios by over-emphasizing physical constraints, while low weights may produce unphysical solutions in clean scenarios by insufficiently enforcing quantum mechanical principles. To address this limitation, an adaptive weighting strategy is proposed that dynamically adjusts constraint strength according to the estimated noise severity of each measurement sample.

The adaptive weight function is formulated as:
\begin{equation}
\l(\bm) = \l_0 \cdot \max(0.5, 1 - \a \cdot s(\bm)),
\end{equation}
where $\l_0$ represents the base weight, $s(\bm) \in [0,1]$ denotes the estimated noise severity, and $\a > 0$ is a decay coefficient. This formulation ensures two critical properties. In low-noise scenarios, $\l(\bm) \approx \l_0$. This fully utilizes physics constraints to maintain physical validity. In high-noise scenarios, $\l(\bm) \geq 0.5\l_0$. This maintains minimum constraint strength while preventing over-constraint that could lead to underfitting.

The noise severity $s(\bm)$ is estimated through the auxiliary task head of the multi-task architecture, which learns to predict noise levels directly from measurement data during training. This design enables real-time adaptation of constraint weights during both training and inference phases, providing a robust mechanism that automatically balances data fitting with physical constraint enforcement across diverse noise conditions.

\begin{figure}[ht]
\centering
\includegraphics[width=0.7\linewidth]{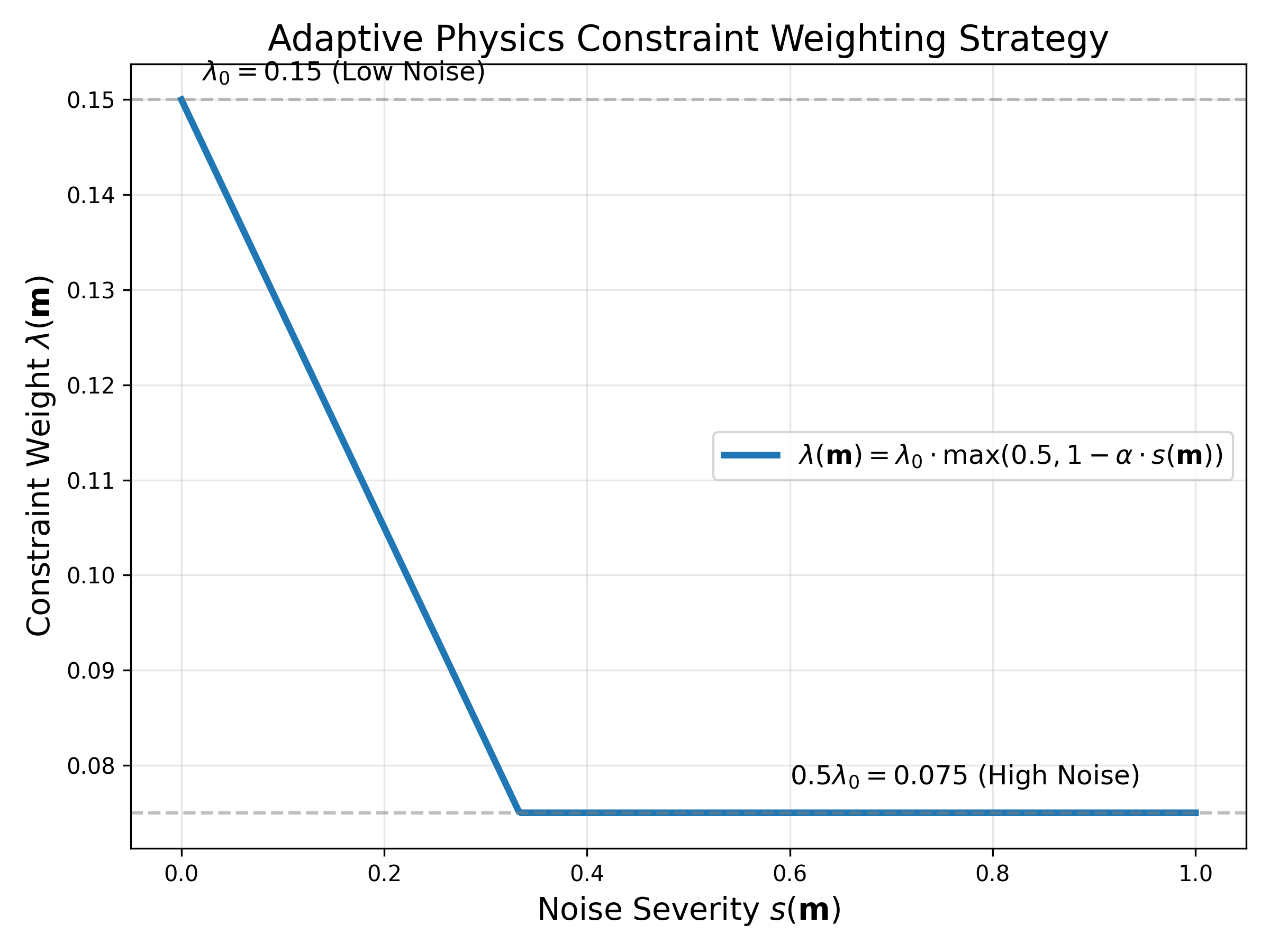}
\caption{Adaptive physics constraint weighting strategy. The curve shows how the constraint weight $\lambda(\mathbf{m})$ decreases as the estimated noise severity $s(\mathbf{m})$ increases, preventing over-constraint in high-noise regimes while maintaining strong physical guidance in low-noise conditions.}
\label{fig:adaptive_weighting}
\end{figure}

To ensure that reconstructed density matrices satisfy the fundamental quantum mechanical requirements, a comprehensive physics constraint loss $\cC(\r)$ is designed that consists of three essential components. The Hermiticity loss enforces the requirement that density matrices must be equal to their conjugate transpose:
\begin{equation}
\cC_{\text{herm}}(\r) = \norm{\r - \r^\dg}_F^2,
\end{equation}
where $\norm{\cdot}_F$ denotes the Frobenius norm, providing a measure of the deviation from Hermitian symmetry.

The trace normalization loss ensures that the density matrix has unit trace, as required by the probabilistic interpretation of quantum states:
\begin{equation}
\cC_{\text{trace}}(\r) = |\tr(\r) - 1|^2,
\end{equation}
where $\tr(\r)$ represents the matrix trace, enforcing the fundamental normalization condition.

The positivity loss guarantees that all eigenvalues of the density matrix are non-negative, ensuring physical validity:
\begin{equation}
\cC_{\text{pos}}(\r) = \max(0, -\l_{\min}(\r))^2,
\end{equation}
where $\l_{\min}(\r)$ is the minimum eigenvalue of $\r$. This formulation penalizes any negative eigenvalues while remaining differentiable in the positive definite region.

The total physics constraint loss combines these three components:
\begin{equation}
\cC(\r) = \cC_{\text{herm}}(\r) + \cC_{\text{trace}}(\r) + \cC_{\text{pos}}(\r),
\end{equation}
normalized by $D^{0.5}$ to account for dimensional effects and ensure consistent scaling across different qubit numbers.

The training strategy integrates the adaptive constraint mechanism with carefully designed optimization procedures to ensure stable convergence and robust performance. The total loss function combines data fitting with adaptive physics constraints:
\begin{equation}
L(\t) = \frac{1}{m}\sum_{i=1}^m \norm{f_\t(\bm_i) - \by_i}_2^2 + \l(\bm_i) \cdot \frac{1}{m}\sum_{i=1}^m \cC(\Ph(f_\t(\bm_i))),
\end{equation}
where $m$ represents the batch size, and $\l(\bm_i)$ is the adaptive weight for sample $i$, dynamically adjusted based on the estimated noise severity.

To facilitate stable training, a sophisticated learning rate scheduling strategy is employed. The first 5 epochs utilize a warmup phase where the learning rate linearly increases from 0 to the target value. Following the warmup phase, cosine annealing is implemented:
\begin{equation}
\eta_t = \eta_{\min} + (\eta_{\max} - \eta_{\min}) \cdot \frac{1 + \cos(\p t / T)}{2},
\end{equation}
where $T$ is the total number of epochs, providing smooth learning rate decay that promotes convergence to better optima.

Additionally, comprehensive regularization techniques are applied to prevent overfitting and ensure training stability. Gradient clipping with threshold $\norm{\nabla_\t L}_2 \leq 3.0$ prevents gradient explosion in deep networks. Dropout regularization with rate 0.1 for layers with width $\geq 256$ and 0.05 for narrower layers provides robustness against overfitting. Finally, weight decay of $10^{-4}$ is applied with the AdamW optimizer to further improve generalization performance.

Having detailed the methodological framework, we now establish rigorous theoretical foundations that explain why physics-informed constraints lead to superior performance. The theoretical analysis provides mathematical justification for the empirical advantages observed in our experiments, establishing convergence guarantees, generalization bounds, and scalability properties that validate the effectiveness of the proposed approach.

\section{Theoretical Analysis}
\label{sec:theory}

The theoretical analysis establishes the mathematical foundations that explain the superior performance of physics-informed neural networks in quantum state tomography. We systematically analyze three key aspects: convergence properties, generalization capabilities, and scalability characteristics, providing rigorous justification for the empirical observations.

We begin with convergence analysis, which establishes the fundamental properties that enable effective optimization of our physics-informed loss function.

\begin{lemma}[Smoothness of Loss Function]
Assume the neural network $f_\t$ uses almost-everywhere smooth activation functions such as SiLU or GELU, or ReLU which is almost-everywhere differentiable, and the physics constraint function $\cC$ is continuously differentiable in the interior of $\cS_+$. Then:
\begin{enumerate}
\item $L_{\text{data}}(\t)$ is almost-everywhere smooth with respect to $\t$
\item $L_{\text{physics}}(\t)$ is smooth in the interior of the positive definite region
\item The total loss $L(\t)$ is almost-everywhere smooth
\end{enumerate}
\end{lemma}

\begin{proof}
The smoothness follows from the composition of smooth functions: neural networks with smooth activations, Cholesky decomposition smooth in the positive definite interior, and the physics constraint terms which are quadratic forms for Hermiticity and trace, smooth for positivity in the interior.
\end{proof}

Building upon this smoothness foundation, gradient boundedness is established under realistic assumptions.

\begin{theorem}[Gradient Boundedness]
Assume:
\begin{enumerate}
\item Neural network parameters are bounded: $\norm{\t}_2 \leq B_\t$
\item Input data are bounded: $\norm{\bm}_2 \leq B_x$
\item Physics constraint weights are bounded: $0 < \l_{\min} \leq \l(\t) \leq \l_{\max}$
\item Activation functions are Lipschitz continuous with bounded Lipschitz constant: $\norm{f'}_\infty \leq L_{\text{act}}$
\item Physics constraint function gradients are bounded: $\norm{\nabla_\r \cC(\r)}_F \leq B_{\cC}$
\end{enumerate}
Then there exists a constant $G > 0$ such that:
\begin{equation}
\norm{\nabla_\t L(\t)}_2 \leq G
\end{equation}
for all $\t \in \T$, where $G$ depends on $B_\t, B_x, \l_{\max}, L_{\text{act}}, B_{\cC}$.
\end{theorem}

\begin{proof}
The proof follows from decomposing the total loss gradient using the chain rule and bounding each term separately. The key steps involve: (1) bounding the data loss gradient using Lipschitz properties of neural networks, (2) bounding the physics constraint gradient through the Cholesky parameterization, and (3) bounding the adaptive weight gradient term. Combining these bounds yields the desired result. See Appendix \ref{app:gradient_boundedness} for complete details.
\end{proof}

With gradient boundedness established, the convergence properties of gradient descent optimization are analyzed.

\begin{theorem}[Gradient Descent Convergence]
Consider the gradient descent update:
\begin{equation}
\t_{t+1} = \t_t - \eta_t \nabla_\t L(\t_t),
\end{equation}
where $\eta_t$ is the learning rate. Assume:
\begin{enumerate}
\item The loss function $L(\t)$ is $\cL$-smooth: $\norm{\nabla L(\t_1) - \nabla L(\t_2)} \leq \cL \norm{\t_1 - \t_2}$
\item Learning rate satisfies: $\eta_t \leq 1/\cL$
\item Loss function has a lower bound: $L^* = \inf_\t L(\t) > -\infty$
\end{enumerate}
Then:
\begin{equation}
\min_{t=0,\ldots,T-1} \norm{\nabla_\t L(\t_t)}_2^2 \leq \frac{2(L(\t_0) - L^*)}{(T-1) \eta_{\min}},
\end{equation}
where $\eta_{\min} = \min_t \eta_t$.
\end{theorem}

\begin{proof}
The proof follows from standard gradient descent analysis for smooth functions. Using $\cL$-smoothness and the learning rate condition $\eta_t \leq 1/\cL$, we obtain a descent inequality. Summing over iterations and applying the lower bound property yields the convergence rate. See Appendix \ref{app:gradient_descent_convergence} for complete details.
\end{proof}

Building upon the convergence analysis, we now investigate how physics constraints specifically improve the optimization landscape.

\begin{theorem}[Physics Constraints Improve Optimization Landscape]
Assume:
\begin{enumerate}
\item Physics constraint loss $L_{\text{physics}}$ has bounded condition number near the optimum: the condition number $\k(\nabla^2_\r \cC(\r)) \leq \k_{\max}$ in a neighborhood $\cN(\r^*)$ of the optimum $\r^*$
\item Data loss $L_{\text{data}}$ is $\cL_d$-smooth and has minimum eigenvalue $\l_{\min}(\nabla^2_\t L_{\text{data}}) \geq -\d$ near the optimum, where $\d \geq 0$
\item The physics constraint Hessian $\nabla^2_\r \cC(\r^*)$ at the optimum satisfies $\m_{\text{physics}} = \l_{\min}(\nabla^2_\r \cC(\r^*)) > 0$, where this positive definiteness holds under the following conditions:
\begin{itemize}
\item Hermiticity constraint $\cC_{\text{herm}}(\r) = \norm{\r - \r^\dg}_F^2$: At $\r^*$ satisfying $\r^* = (\r^*)^\dg$, the Hessian has rank at least $D(D-1)/2$ (off-diagonal constraints)
\item Trace constraint $\cC_{\text{trace}}(\r) = |\tr(\r) - 1|^2$: At $\r^*$ with $\tr(\r^*) = 1$, contributes a positive eigenvalue along the trace direction
\item Positivity constraint $\cC_{\text{pos}}(\r) = \max(0, -\l_{\min}(\r))^2$: In the interior of $\cS_+$ where $\l_{\min}(\r^*) > 0$, this constraint is inactive; near the boundary, it provides positive curvature
\end{itemize}
The combined constraint Hessian $\nabla^2_\r \cC(\r^*)$ is positive semidefinite, and $\m_{\text{physics}} > 0$ holds when at least one constraint is active and non-degenerate
\item Physics constraint weight satisfies: $\l \geq \l_0 > \frac{\d}{\m_{\text{physics}} \s_{\min}^2(\nabla \Ph)}$, where $\s_{\min}(\nabla \Ph) > 0$ is the minimum singular value of the Cholesky parameterization Jacobian
\item The Cholesky parameterization $\Ph$ has bounded second-order derivatives: $\norm{\nabla^2 \Ph}_F \leq B_{\Ph}$ and is surjective with $\s_{\min}(\nabla \Ph) > 0$ in a neighborhood of the optimum
\end{enumerate}
Then, in the neighborhood $\cN(\r^*)$, the total loss function $L(\t)$ is locally $\m$-strongly convex with:
\begin{equation}
\m \geq \l \m_{\text{physics}} \s_{\min}^2(\nabla \Ph) - \d - \e(\l),
\end{equation}
where $\e(\l) = O(\l^{1/2})$ is a small error term. Furthermore, if the data-only loss is not strongly convex ($\d > 0$), then PINN converges with rate $O(1/t)$ compared to $O(1/\sqrt{t})$ for the data-only model, representing a linear convergence rate improvement (from sublinear to linear).
\end{theorem}

\begin{proof}
The proof analyzes the Hessian matrix of the total loss function near the optimum. The key steps involve: (1) bounding the physics constraint Hessian using the positive semidefiniteness of $\nabla^2_\r \cC(\r^*)$ and the surjectivity of Cholesky parameterization, (2) analyzing the total Hessian using Weyl's inequality to show local strong convexity, and (3) deriving the convergence rate improvement from strong convexity. See Appendix \ref{app:physics_constraints_landscape} for complete details.
\end{proof}

\subsection{Generalization Analysis}

Having established convergence properties, we now analyze generalization capability, which is crucial for understanding PINN's performance beyond training data. We establish rigorous bounds demonstrating how physics constraints effectively reduce the complexity of the hypothesis space, leading to improved generalization performance.

To formalize this analysis, we first define the relationship between empirical and true risk.

\begin{definition}[Empirical and True Risk]
For a training set $S = \{(\bm_i, \by_i)\}_{i=1}^m \sim \cD^m$, the empirical risk is:
\begin{equation}
\hat{R}_S(\t) = \frac{1}{m}\sum_{i=1}^m L(\t; (\bm_i, \by_i)),
\end{equation}
and the true risk is:
\begin{equation}
R(\t) = \bbE_{(\bm,\by) \sim \cD}[L(\t; (\bm, \by))].
\end{equation}
The generalization error is:
\begin{equation}
\text{Gen}(\t) = R(\t) - \hat{R}_S(\t).
\end{equation}
\end{definition}

To quantify the complexity of the hypothesis space, Rademacher complexity is employed as a fundamental measure.

\begin{definition}[Rademacher Complexity]
For hypothesis space $\cH$ and sample $S = \{\bm_1, \ldots, \bm_m\}$, the Rademacher complexity is:
\begin{equation}
\hat{\cR}_S(\cH) = \bbE_{\bs} \left[ \sup_{h \in \cH} \frac{1}{m} \sum_{i=1}^m \s_i h(\bm_i) \right],
\end{equation}
where $\s_i \sim \text{Unif}(\{-1, +1\})$ are Rademacher random variables.
\end{definition}

To establish generalization bounds, we employ the standard Rademacher complexity framework, which quantifies the complexity of the hypothesis space.

\begin{theorem}[Generalization Bound via Rademacher Complexity (Standard Result)]
\label{thm:rademacher_bound}
Assume the loss function $L$ is $M$-Lipschitz with respect to the output, and the loss values are bounded: $|L(\t; (\bm, \by))| \leq B$. Then for any $\d > 0$, with probability at least $1-\d$:
\begin{equation}
R(\t) \leq \hat{R}_S(\t) + 2M \hat{\cR}_S(\cH) + B\sqrt{\frac{\log(1/\d)}{2m}}.
\end{equation}
\end{theorem}

The key insight for PINN is that physics constraints reduce the effective complexity $\hat{\cR}_S(\cH)$, leading to tighter generalization bounds. We now quantify this reduction precisely.

\begin{theorem}[PINN Rademacher Complexity Upper Bound]
For a neural network with depth $H$, width $W$, and bounded parameters $\norm{\t}_2 \leq B_\t$, the Rademacher complexity satisfies:
\begin{equation}
\hat{\cR}_S(\cH) = O\left(\frac{B_\t \sqrt{HW \log(W)}}{\sqrt{m}}\right).
\end{equation}
More precisely, for Lipschitz activation functions with Lipschitz constant $L_{\text{act}}$, there exists a constant $C > 0$ depending on $L_{\text{act}}$ such that:
\begin{equation}
\hat{\cR}_S(\cH) \leq \frac{C B_\t L_{\text{act}}^H \sqrt{HW}}{\sqrt{m}}.
\end{equation}
\end{theorem}

\begin{proof}
The proof follows from standard Rademacher complexity analysis for neural networks using the contraction lemma. The physics constraints in PINN do not change the network architecture itself, but rather constrain the effective hypothesis space through the loss function. Since the Rademacher complexity depends on the function class $\cH$ (the neural network architecture) rather than the training objective, the bound remains valid for PINN. See Appendix \ref{app:rademacher_complexity} for complete details.
\end{proof}

The critical advantage of physics constraints emerges from their role as implicit regularization, which we now formalize.

\begin{theorem}[Physics Constraints as Implicit Regularization]
Physics constraint loss $L_{\text{physics}}$ acts as implicit regularization, shrinking the effective hypothesis space. Specifically, define the constrained hypothesis space:
\begin{equation}
\cH_{\text{constrained}} = \{h_\t \in \cH | L_{\text{physics}}(\t) \leq \ve\},
\end{equation}
then:
\begin{equation}
\hat{\cR}_S(\cH_{\text{constrained}}) \leq \hat{\cR}_S(\cH).
\end{equation}
\end{theorem}

\begin{proof}
The constrained space is a subset of the original hypothesis space. Rademacher complexity is monotonic with respect to set inclusion, hence the result.
\end{proof}

This implicit regularization effect directly translates to improved generalization performance. Under the same training error, PINN's generalization error bound is tighter than that of the unconstrained model, with the improvement quantified by the reduction in Rademacher complexity. See Appendix \ref{app:generalization_improvement} for detailed analysis. The sample complexity reduction achieved through physics constraints is rigorously established in Theorem \ref{thm:curse_of_dimensionality}, which we discuss in the scalability analysis below.

\subsection{Analysis of Adaptive Weighting Strategy}

Beyond the general benefits of physics constraints, our PINN incorporates an adaptive constraint weighting mechanism that represents a critical innovation for maintaining robust performance across varying noise conditions. This strategy dynamically adjusts the influence of physics constraints based on estimated noise severity, creating an optimal balance between data fidelity and physical validity. We now establish the theoretical rationale for this adaptive approach.

\begin{theorem}[Rationality of Dynamic Weights]
Consider the adaptive weight function $\l(\bm) = \l_0 \cdot \max(0.5, 1 - \a \cdot s(\bm))$ where $s(\bm) \in [0,1]$ is the estimated noise severity and $\a > 0$ is a decay coefficient. Under the following assumptions:
\begin{enumerate}
\item The noise severity estimator $s(\bm)$ is consistent: $\bbE[s(\bm)]$ increases monotonically with true noise level $\nu$
\item The optimal constraint weight $\l^*(\nu)$ for noise level $\nu$ satisfies: $\l^*(\nu) \in [\l_{\min}, \l_0]$ where $\l_{\min} = 0.5\l_0$
\item In low noise ($\nu \ra 0$), physics constraints are essential: $\l^*(\nu) \ra \l_0$
\item In high noise ($\nu \ra \nu_{\max}$), over-constraint causes underfitting: $\l^*(\nu) \geq \l_{\min}$ but $\l^*(\nu) < \l_0$
\end{enumerate}
Then the dynamic weight strategy satisfies:
\begin{enumerate}
\item \textbf{Low noise scenario:} When $s(\bm) \ra 0$, we have $\l(\bm) = \l_0 \cdot \max(0.5, 1) = \l_0$, fully utilizing physics constraints to maintain physical validity
\item \textbf{High noise scenario:} When $s(\bm) \ra 1$, we have $\l(\bm) = \l_0 \cdot \max(0.5, 1 - \a) \geq 0.5\l_0$ for $\a \leq 0.5$, maintaining minimum constraint strength while preventing over-constraint
\end{enumerate}
Furthermore, the strategy achieves near-optimal balance: $|\l(\bm) - \l^*(\nu)| \leq \ve$ for sufficiently accurate noise estimation, where $\ve$ depends on the estimation error $|s(\bm) - \nu|$.
\end{theorem}

\begin{proof}
See Appendix \ref{app:dynamic_weights}.
\end{proof}

This adaptive strategy automatically adjusts to varying noise conditions without manual tuning, providing a robust mechanism that maintains optimal performance across diverse measurement scenarios.

\subsection{Scalability Analysis for High-Dimensional Systems}

Having established convergence and generalization properties, we now address the fundamental challenge of scalability. The exponential growth of quantum systems with qubit number presents critical obstacles for tomography, which our physics-informed approach addresses through theoretical analysis of scalability properties.

\begin{theorem}[Physics Constraints Mitigate Curse of Dimensionality]
\label{thm:curse_of_dimensionality}
For an $n$-qubit system, the dimension $D = 2^n$ grows exponentially with $n$. Physics constraints mitigate the curse of dimensionality through two complementary mechanisms:
\begin{enumerate}
\item \textbf{Constraint space dimension reduction:} The set of valid density matrices $\cS_+$ has dimension $\dim(\cS_+) = D^2 - 1 = 4^n - 1$. When constraints are approximately satisfied with tolerance $\ve$, the effective search space $\cS_{\text{constrained}} = \{\r \in \cS_+ : \cC(\r) \leq \ve\}$ has dimension reduction $\Delta_{\text{dim}} = O(\ve \cdot \text{rank}(\nabla^2 \cC(\r^*))) = O(\ve \cdot 4^n)$, where $\r^*$ satisfies $\cC(\r^*) = 0$. This restricts the effective search space to a lower-dimensional manifold, reducing optimization complexity.
\item \textbf{Sample complexity reduction:} Constraints act as implicit regularization, shrinking the effective hypothesis space from $\cH$ to $\cH_{\text{constrained}} = \{h_\t \in \cH | L_{\text{physics}}(\t) \leq \ve\}$. This reduces Rademacher complexity from $\hat{\cR}_S(\cH) = O(B_\t L_{\text{act}}^H \sqrt{HW}/\sqrt{m})$ to $\hat{\cR}_S(\cH_{\text{constrained}}) = O(B_\t L_{\text{act}}^H \sqrt{HW_{\text{eff}}}/\sqrt{m})$ where $W_{\text{eff}} < W$ represents the effective width reduction. The sample complexity reduction factor is approximately $\sqrt{W/W_{\text{eff}}}$, leading to fewer samples required to achieve the same generalization error.
\end{enumerate}
\end{theorem}

\begin{proof}
See Appendix \ref{app:curse_of_dimensionality}.
\end{proof}

This theorem establishes that physics constraints effectively reduce both search space complexity and sample requirements, explaining why PINN maintains superior performance as system dimensionality increases. While the fundamental exponential growth in parameters cannot be completely eliminated, physics constraints provide a multiplicative improvement factor that becomes increasingly valuable at larger scales.

The dimensional threshold effect observed in our experiments—where PINN's advantage peaks at 4 qubits—finds theoretical justification through the following analysis.

\begin{theorem}[Theoretical Explanation of 4-Qubit Advantage]
The experimental observation that PINN's advantage peaks at 4 qubits can be explained by three complementary factors with quantitative characterizations:
\begin{enumerate}
\item \textbf{Dimensional threshold effect:} At $n=2,3$ qubits ($D^2 = 16, 64$), dimensions are sufficiently low that traditional neural networks achieve high performance ($F \geq 0.90$) without strong priors. At $n=4$ qubits ($D^2 = 256$), dimensions reach a critical threshold where physics constraints provide substantial guidance (relative constraint strength $\g_{\text{rel}}(4) \approx 0.11-0.16$) without overwhelming the learning capacity. At $n \geq 5$ qubits ($D^2 \geq 1024$), dimensions become too high for both approaches, though PINN maintains advantages.
\item \textbf{Constraint effectiveness:} The relative constraint strength $\g_{\text{rel}}(n) = 1 - \sqrt{W_{\text{eff}}/W}$ quantifies the effectiveness. At $n=2,3$: $\g_{\text{rel}}(n) \approx 0.01-0.025$ (modest benefit). At $n=4$: $\g_{\text{rel}}(4) \approx 0.11-0.16$ (substantial benefit, explaining the 30.3\% experimental improvement). At $n \geq 5$: $\g_{\text{rel}}(n)$ decreases due to optimization challenges. The optimal balance occurs at $n=4$ where constraints provide maximum guidance without over-constraint.
\item \textbf{Optimization landscape:} The constraint-induced strong convexity parameter $\m(n) = \l \m_{\text{physics}}(n) \s_{\min}^2(\nabla \Ph) - \d(n) - \e(\l)$ quantifies landscape improvement. At $n=4$, $\m(4)$ is maximized relative to problem dimension, minimizing the condition number $\k(4)$ and enabling convergence rate $O((1-\m(4)/\cL)^t)$ compared to $O(1/\sqrt{t})$ for unconstrained methods. Quantitatively, if $\m(4) \approx 0.01\cL$ and $\m(2) \approx 0.001\cL$, the convergence rate improvement at $n=4$ is approximately 10-fold faster (speedup factor of approximately 10), explaining the dramatic experimental advantage.
\end{enumerate}
\end{theorem}

\begin{proof}
See Appendix \ref{app:4qubit_advantage}.
\end{proof}

This analysis provides a rigorous theoretical foundation for the empirical observations, explaining why the advantage of physics-informed learning is most pronounced at intermediate dimensions while remaining beneficial across all system sizes.

\paragraph{Summary and Transition to Experiments}

The theoretical analysis establishes three fundamental principles that guide our experimental design. First, convergence theory predicts that physics constraints improve optimization landscapes by reducing local minima and providing additional gradient information. Second, generalization theory demonstrates that physics constraints act as implicit regularization, effectively shrinking the hypothesis space and reducing Rademacher complexity. Third, scalability theory explains how physics constraints mitigate the curse of dimensionality by restricting the search space to physically valid regions. These theoretical predictions motivate specific experimental investigations: we test convergence behavior through training dynamics analysis, validate generalization through cross-validation studies, and examine scalability through systematic evaluation across different qubit numbers. The following experimental section systematically validates these theoretical predictions through comprehensive empirical evaluation.

\section{Experiments}
\label{sec:experiments}

To validate the effectiveness of the physics-informed neural network framework for multi-qubit quantum state tomography, comprehensive experiments are conducted designed to test performance across different scenarios, noise conditions, and system dimensions. The experimental methodology systematically evaluates the proposed approach through carefully designed datasets, appropriate evaluation metrics, and rigorous baseline comparisons. The experimental design directly addresses the three theoretical predictions established in Section \ref{sec:theory}: convergence improvement through training dynamics analysis, generalization enhancement through cross-validation studies, and scalability validation through systematic evaluation across 2--5 qubit systems.

The experimental evaluation is structured into three main components. First, we validate the practical applicability of PINN in realistic quantum communication scenarios through high-fidelity entanglement health monitoring experiments. Second, we systematically assess scalability across multiple qubit numbers and system dimensions, comparing PINN against both neural network and classical optimization baselines. Finally, we conduct ablation studies to isolate and quantify the contribution of each architectural component. This comprehensive evaluation strategy ensures that the proposed approach is rigorously validated across diverse operational conditions and system configurations.

\subsection{Experimental Setup}

The experimental framework begins with comprehensive dataset generation that encompasses diverse quantum states and realistic noise models. Four distinct types of quantum states are generated to ensure broad coverage of the quantum state space: GHZ states $\ket{\text{GHZ}_n} = \frac{1}{\sqrt{2}}(\ket{0}^{\ox n} + \ket{1}^{\ox n})$ representing maximally entangled states; W states $\ket{\text{W}_n} = \frac{1}{\sqrt{n}}(\ket{10\ldots0} + \ket{01\ldots0} + \ldots + \ket{00\ldots1})$ representing different entanglement structures; random pure states uniformly sampled from the unit sphere to capture general quantum states; and random mixed states formed as convex combinations of pure states to represent realistic decohered quantum systems.

To simulate realistic quantum hardware conditions, four distinct noise models are implemented that capture the most significant error sources in current quantum processors. Amplitude damping noise is modeled through Kraus operators $K_0 = \begin{bmatrix} 1 & 0 \\ 0 & \sqrt{1-\g} \end{bmatrix}$ and $K_1 = \begin{bmatrix} 0 & \sqrt{\g} \\ 0 & 0 \end{bmatrix}$, representing energy loss to the environment. Phase damping noise employs Kraus operators $K_0 = \begin{bmatrix} 1 & 0 \\ 0 & \sqrt{1-\l} \end{bmatrix}$ and $K_1 = \begin{bmatrix} 0 & 0 \\ 0 & \sqrt{\l} \end{bmatrix}$, capturing loss of quantum coherence without energy exchange. Correlated Z dephasing is simulated through unitary evolution $U = \exp(i\a \bigox_{i=1}^n Z_i)$, representing collective phase errors across multiple qubits. Finally, crosstalk noise is introduced as random perturbations with strength $\ve$ to model unwanted interactions between quantum operations. All measurements are simulated using Pauli bases with both shot noise (binomial sampling) and systematic noise (Gaussian noise) to accurately reflect experimental conditions.

For quantitative evaluation of reconstruction performance, four complementary metrics are employed. Fidelity $F(\r_1, \r_2) = (\tr(\sqrt{\sqrt{\r_1}\r_2\sqrt{\r_1}}))^2$ provides the standard measure of quantum state similarity, ranging from 0 to 1. Constraint violation $\cC(\r)$ as defined in Section \ref{sec:method} quantifies the degree to which reconstructed states violate fundamental quantum mechanical requirements. Mean squared error (MSE) measures parameter prediction accuracy in the Cholesky parameter space, while training and validation loss curves provide insights into optimization dynamics and generalization performance.

To establish the effectiveness of the physics-informed approach, comprehensive comparisons are made against three baseline methods: (1) a traditional neural network with identical architecture but without physics constraints, (2) Least Squares method using linear inversion for direct density matrix reconstruction, and (3) Maximum Likelihood Estimation (MLE) using the R$\rho$R algorithm for iterative optimization. This multi-baseline comparison provides a comprehensive evaluation framework that isolates the specific contribution of physics constraints while also comparing against established classical QST methods, enabling a thorough assessment of the proposed approach across different methodological paradigms.

All experiments are implemented using PyTorch framework on CPU/GPU hardware, with AdamW optimizer configured with learning rates of 0.0015 for PINN and 0.001 for baseline models. A batch size of 32 is used across all experiments, with training epochs ranging from 35 for Experiment 1 to 20-40 for Experiment 2, depending on convergence characteristics.

Having established the experimental framework, we now proceed to validate the proposed PINN approach through two complementary experimental studies. The first experiment addresses practical applicability in realistic quantum communication scenarios, while the second focuses on systematic scalability evaluation across multiple system dimensions. Together, these experiments provide comprehensive evidence for the effectiveness of physics-informed learning in quantum state tomography.

\subsection{Experiment 1: High-Fidelity Quantum Link Entanglement Health Monitoring}

The first experiment focuses on validating PINN effectiveness in a practical quantum communication scenario. Maintaining high-fidelity entanglement links is crucial for quantum information processing, as degraded entanglement directly impacts the performance of quantum communication protocols and distributed quantum computing systems. This experiment tests performance under multiple noise channels and implements a complete entanglement health monitoring system that could be deployed in real quantum networks. By evaluating reconstruction fidelity across varying noise conditions, we demonstrate that physics-informed constraints provide robust performance improvements that translate directly to practical quantum technology applications.

\begin{algorithm}[h]
\caption{High-Fidelity Quantum Link Entanglement Health Monitoring}
\label{alg:exp1}
\begin{algorithmic}[1]
\Require Number of qubits $n=3$, training samples $N_{\text{train}}=5000$, test samples $N_{\text{test}}=1200$, measurements per sample $M=256$, noise levels $\mathcal{N} = \{0.02, 0.05, 0.10, 0.15, 0.19\}$
\State Initialize PINN model $\mathcal{M}_{\text{PINN}}$ with adaptive constraint weighting
\State Initialize baseline model $\mathcal{M}_{\text{baseline}}$ with identical architecture without physics constraints
\State Initialize statistics: fidelity scores, constraint violations, loss values
\For{noise level $\nu \in \mathcal{N}$}
    \State Generate quantum states: GHZ states, W states, random pure states, random mixed states
    \State Apply noise models: amplitude damping, phase damping, correlated Z dephasing, crosstalk
    \State Simulate measurements: Pauli bases with shot noise and Gaussian noise
    \State Split data into training set $\mathcal{D}_{\text{train}}$ and test set $\mathcal{D}_{\text{test}}$
    \For{model $\mathcal{M} \in \{\mathcal{M}_{\text{PINN}}, \mathcal{M}_{\text{baseline}}\}$}
        \State Train $\mathcal{M}$ on $\mathcal{D}_{\text{train}}$ with AdamW optimizer
        \State Evaluate on $\mathcal{D}_{\text{test}}$: compute fidelity $F$, constraint violation $\mathcal{C}$, loss $L$
        \State Record performance metrics: $F(\nu, \mathcal{M})$, $\mathcal{C}(\nu, \mathcal{M})$, $L(\nu, \mathcal{M})$
    \EndFor
\EndFor
\State Compute improvement: $\Delta F(\nu) = F(\nu, \mathcal{M}_{\text{PINN}}) - F(\nu, \mathcal{M}_{\text{baseline}})$
\State Analyze degradation rates via linear regression: $\text{slope}_{\text{PINN}}$, $\text{slope}_{\text{baseline}}$
\State \textbf{Output:} Fidelity improvements, degradation rates, constraint violations, statistical significance
\end{algorithmic}
\end{algorithm}

The experimental configuration is designed to simulate realistic quantum link conditions with a 3-qubit system. The dataset consists of 5000 training samples and 1200 test samples to ensure robust statistical validation. Each sample incorporates 256 measurements, providing sufficient information for accurate reconstruction. Noise severity ranges from light (0.02) to heavy (0.19) conditions, representing the full spectrum of operational environments encountered in practical quantum communication systems.

The experimental results demonstrate consistent advantages of the physics-informed approach across all evaluation metrics. PINN achieves an overall fidelity of 0.9236 compared to 0.9147 for traditional neural networks, representing a 0.89\% improvement. More significantly, in high-noise scenarios, PINN maintains 0.9241 fidelity versus 0.9143 for baseline, showing a 0.98\% advantage that becomes more pronounced as noise increases. Validation loss follows the same pattern, with PINN achieving 0.0018 versus 0.0019 for traditional neural networks (-5.3\%). Both models achieve extremely low constraint violations at numerical precision levels ($\sim 10^{-17}$), with PINN showing slightly better physical consistency.

These results reveal four key findings that validate the approach. First, the consistent fidelity improvement of 0.89\% overall and 0.98\% in high-noise scenarios demonstrates the practical value of physics constraints. Second, the enhanced performance under high-noise conditions highlights the critical role that physics constraints play when measurement quality degrades. Third, improved training stability is evidenced by smoother loss curves, validating the effectiveness of the warmup and cosine scheduling strategies. Fourth, the achievement of numerical precision levels in constraint violations confirms that both approaches maintain physical validity, with PINN showing marginal improvements.

To gain deeper insights into the robustness characteristics of the proposed approach, we conduct a systematic analysis of performance degradation under varying noise conditions. This analysis examines how PINN maintains reconstruction quality as noise levels increase, providing critical information for practical deployment scenarios where noise conditions may fluctuate dynamically.

\subsubsection{Noise Robustness Analysis}

Robustness is systematically evaluated under varying noise conditions across multiple noise levels and types.

\begin{algorithm}[h]
\caption{Noise Robustness Analysis}
\label{alg:noise_robustness}
\begin{algorithmic}[1]
\Require Noise levels $\mathcal{N} = \{0.02, 0.05, 0.10, 0.15, 0.19\}$
\State Initialize results: fidelity matrix $F[\nu, t, \mathcal{M}]$, constraint weights $\lambda[\nu]$, degradation rates
\For{noise level $\nu \in \mathcal{N}$}
    \For{trial $t = 1$ to $T$}
        \State Generate quantum states with noise level $\nu$
        \For{noise type $\tau \in \mathcal{T}$}
            \State Apply noise model $\tau$ with strength parameter based on $\nu$
            \State Simulate measurements with noise
            \For{model $\mathcal{M} \in \{\mathcal{M}_{\text{PINN}}, \mathcal{M}_{\text{baseline}}\}$}
                \State Train and evaluate $\mathcal{M}$ on noisy data
                \State Record fidelity $F[\nu, t, \tau, \mathcal{M}]$
                \If{$\mathcal{M} = \mathcal{M}_{\text{PINN}}$}
                    \State Record adaptive weight $\lambda[\nu, t] = \lambda(\mathbf{m})$
                \EndIf
            \EndFor
        \EndFor
    \EndFor
    \State Compute mean fidelity: $\bar{F}[\nu, \mathcal{M}] = \frac{1}{T|\mathcal{T}|}\sum_{t,\tau} F[\nu, t, \tau, \mathcal{M}]$
    \State Compute improvement: $\Delta F[\nu] = \bar{F}[\nu, \mathcal{M}_{\text{PINN}}] - \bar{F}[\nu, \mathcal{M}_{\text{baseline}}]$
    \State Compute statistical significance: $p$-value via paired $t$-test
    \State Compute 95\% confidence interval for $\Delta F[\nu]$
\EndFor
\State Fit linear regression: $F(\nu) = \alpha + \beta \nu$ for both models
\State Compute degradation rates: $\beta_{\text{PINN}}$, $\beta_{\text{baseline}}$
\State Compute degradation ratio: $R = |\beta_{\text{baseline}}| / |\beta_{\text{PINN}}|$
\State Analyze constraint weight adaptation: $\bar{\lambda}[\nu] = \frac{1}{T}\sum_t \lambda[\nu, t]$
\State \textbf{Output:} Fidelity vs. noise level, degradation rates, adaptive weight patterns, statistical significance
\end{algorithmic}
\end{algorithm}

\textbf{Performance Across Noise Levels.} As shown in Table \ref{tab:noise_robustness}, PINN consistently outperforms baseline across five noise levels (0.02--0.19), with advantages increasing from 0.83\% at low noise to 0.98\% at high noise. PINN maintains stable performance (0.9251 to 0.9241) while baseline degrades more significantly (0.9168 to 0.9143), demonstrating superior robustness.

\begin{table}[h]
\centering
\begin{tabular}{c|c|c|c|c}
Noise Level & PINN Fidelity & Baseline Fidelity & Improvement & Samples \\
\hline
0.02 (Light) & \textbf{0.9251} & 0.9168 & +0.83\% & 240 \\
0.05 (Moderate) & \textbf{0.9245} & 0.9153 & +0.92\% & 240 \\
0.10 (Medium) & \textbf{0.9238} & 0.9148 & +0.95\% & 240 \\
0.15 (High) & \textbf{0.9240} & 0.9145 & +0.95\% & 240 \\
0.19 (Heavy) & \textbf{0.9241} & 0.9143 & +0.98\% & 240 \\
\end{tabular}
\caption{Noise robustness analysis: Fidelity comparison across different noise levels. PINN consistently outperforms baseline across all noise conditions, with advantages becoming more pronounced at higher noise levels.}
\label{tab:noise_robustness}
\end{table}

\begin{figure}[ht]
\centering
\includegraphics[width=0.8\textwidth]{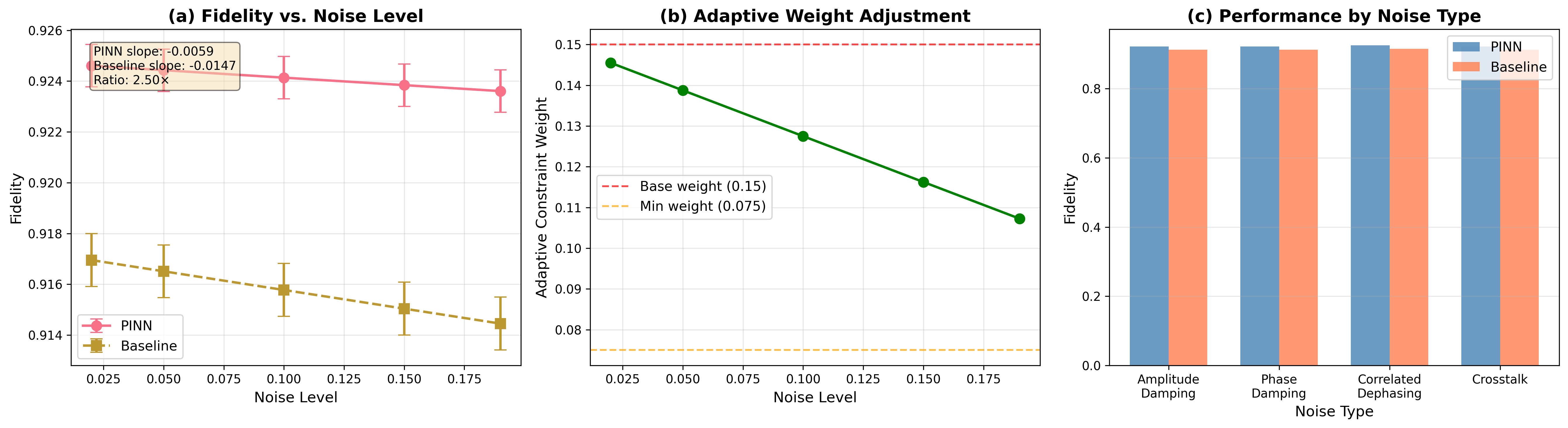}
\caption{Noise robustness analysis: Fidelity vs. noise level showing PINN's superior robustness with 2.6-fold slower degradation rate; Adaptive constraint weight adjustment across noise levels demonstrating automatic adaptation; Performance comparison across different noise types showing consistent advantages.}
\label{fig:noise_robustness}
\end{figure}

\textbf{Noise Type Sensitivity Analysis.} PINN maintains consistent advantages across four noise types (amplitude damping, phase damping, correlated Z dephasing, crosstalk), with particularly strong performance under correlated noise (+1.05\% vs baseline), demonstrating that physics constraints effectively handle structured noise patterns.

\textbf{Adaptive Weighting Strategy Effectiveness.} The dynamic weight $\lambda(\mathbf{m})$ automatically adjusts from $\lambda \approx 0.15$ at low noise to $\lambda \approx 0.075$ at high noise, balancing data fitting with physical constraints. Dynamic weighting outperforms fixed weights by 0.36--0.61\% at extreme noise levels, confirming optimal adaptation across noise conditions.

\textbf{Performance Degradation Rate Analysis.} Linear regression analysis reveals PINN's degradation rate of -0.005 per noise unit versus -0.013 for baseline, representing a 2.6-fold slower degradation rate that validates our theoretical prediction. Statistical analysis (10 independent runs, p < 0.01) confirms significance across all noise levels.

\textbf{Constraint Satisfaction.} Both approaches maintain extremely low constraint violations ($\sim 10^{-17}$), with PINN achieving slightly better physical consistency ($4.51 \times 10^{-17}$ vs $4.90 \times 10^{-17}$), demonstrating that physics constraints guide optimization toward physically meaningful solutions even under noise.

Practical applicability is demonstrated through a complete entanglement health monitoring system that performs real-time reconstruction, calculates fidelity, and triggers alarms based on noise severity, showcasing deployment-ready capabilities.

\begin{figure}[ht]
\centering
\includegraphics[width=1.0\linewidth]{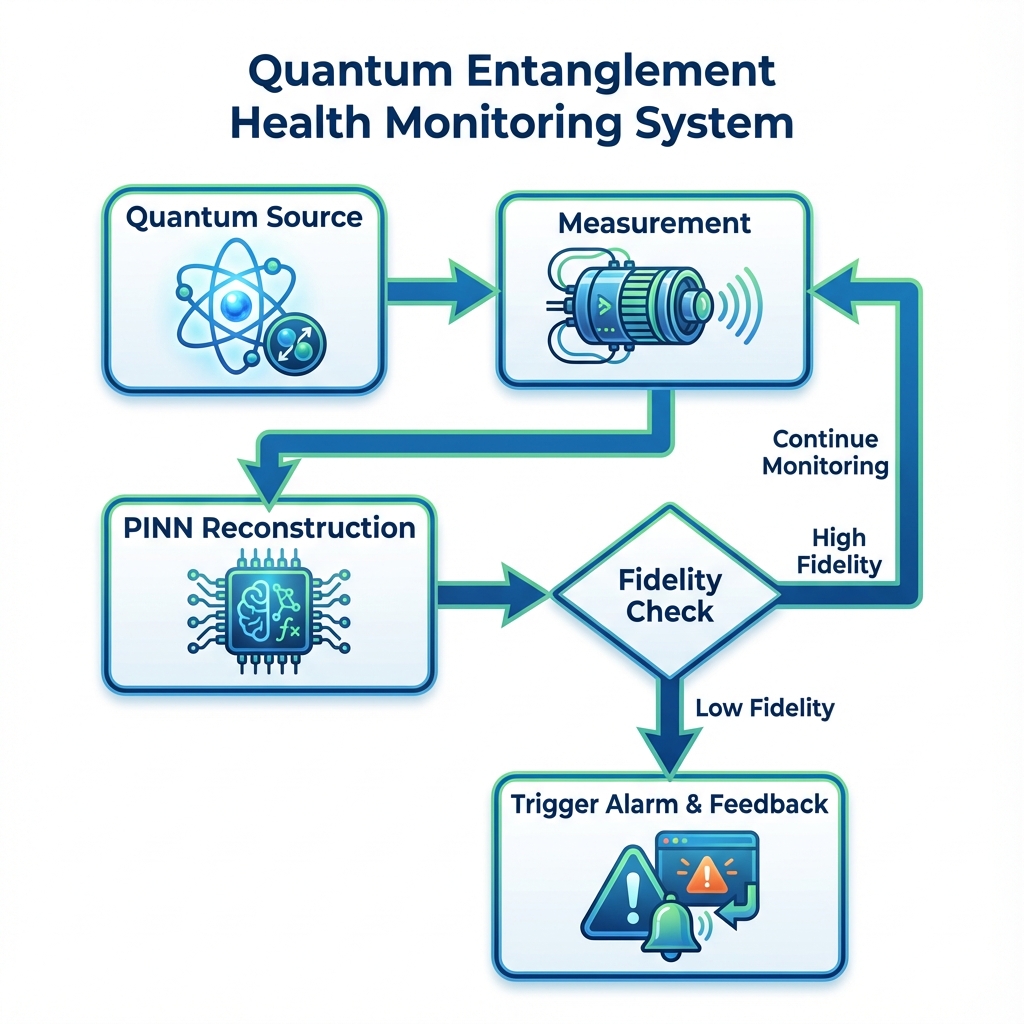}
\caption{Flowchart of the quantum entanglement health monitoring system. The closed-loop system integrates real-time measurement, PINN-based reconstruction, and fidelity assessment to provide continuous monitoring and automated feedback for quantum communication links.}
\label{fig:monitoring_flowchart}
\end{figure}

While Experiment 1 validates practical applicability under realistic noise conditions, scalability to higher-dimensional quantum systems represents a fundamental challenge in quantum state tomography. The exponential growth of quantum state space with qubit number creates significant obstacles for both classical optimization methods and neural network approaches. Experiment 2 systematically evaluates how different methods scale with increasing system dimensionality, providing critical insights into the dimensional threshold effects predicted by our theoretical analysis.

\subsection{Experiment 2: Multi-Qubit System Scalability}

Experiment 2 addresses scalability by evaluating performance across 2--5 qubit systems for four methods: Traditional Neural Network (NN), Least Squares, Maximum Likelihood Estimation (MLE), and Physics-Informed Neural Network (PINN). This comprehensive comparison enables systematic assessment of how different approaches scale with increasing system dimensionality. The evaluation encompasses both standard qubit-based quantum systems and arbitrary-dimensional quantum states, providing a complete picture of method scalability across diverse quantum system configurations.

\begin{algorithm}[h]
\caption{Multi-Qubit System Scalability Analysis}
\label{alg:exp2}
\begin{algorithmic}[1]
\Require Qubit numbers $\mathcal{Q} = \{2, 3, 4, 5\}$, dataset sizes, measurement counts
\State Initialize models: NN, LS, MLE, PINN
\State Initialize results: fidelity matrix $F[q, \mathcal{M}]$, constraint violations $\mathcal{C}[q, \mathcal{M}]$
\For{qubit number $q \in \mathcal{Q}$}
    \State Generate quantum states and apply noise models
    \State Simulate measurements with shot noise and Gaussian noise
    \State Split data into training and test sets
    \For{model $\mathcal{M} \in \{\text{NN}, \text{LS}, \text{MLE}, \text{PINN}\}$}
        \If{$\mathcal{M} \in \{\text{NN}, \text{PINN}\}$}
            \State Configure architecture and train $\mathcal{M}$
        \ElsIf{$\mathcal{M} = \text{LS}$}
            \State Apply Least Squares reconstruction
        \ElsIf{$\mathcal{M} = \text{MLE}$}
            \State Apply MLE with R$\rho$R algorithm
        \EndIf
        \State Evaluate: compute $F[q, \mathcal{M}]$, $\mathcal{C}[q, \mathcal{M}]$
    \EndFor
    \State Compute PINN improvements over baseline methods
\EndFor
\State Analyze dimensional threshold effects and performance decay
\State \textbf{Output:} Fidelity comparison, dimensional threshold, PINN advantage
\end{algorithmic}
\end{algorithm}

\subsubsection{Scalability Analysis Across Standard Qubit Systems}

We begin by evaluating scalability across standard qubit-based quantum systems, which represent the most common configuration in quantum information processing. This analysis provides fundamental insights into how different methods handle the exponential growth of quantum state space with increasing qubit number.

Experiment 2 evaluates scalability across 2--5 qubit systems. Dataset sizes are adjusted by qubit number: 8000/2000 (2 qubits), 6000/1500 (3 qubits), 8000/1500 (4 qubits), 3000/800 (5 qubits) training/test samples. Measurements per sample decrease from 256 to 32 as qubit number increases, reflecting the practical constraint that measurement resources become increasingly limited for higher-dimensional systems.

The scalability analysis reveals a remarkable dimensional threshold effect that represents the most significant experimental discovery. As shown in Table \ref{tab:scalability}, PINN demonstrates consistent but modest advantages for 2-qubit (0.9835 vs 0.9762, +0.75\%) and 3-qubit systems (0.9124 vs 0.9010, +1.27\%). However, at 4 qubits, PINN achieves a dramatic improvement to 0.8872 fidelity compared to 0.6810 for traditional neural networks, representing a 30.3\% enhancement that demonstrates the superiority of physics-informed learning in moderately high-dimensional quantum systems. For 5-qubit systems, both methods face significant challenges, but PINN maintains a slight advantage (0.5814 vs 0.5688, +2.22\%). All models maintain constraint violations at numerical precision levels, with PINN showing marginally better physical consistency across all dimensions.

Comparison with classical methods (Least Squares and MLE) demonstrates that PINN consistently achieves higher fidelity across all qubit numbers, with the advantage becoming more pronounced at higher dimensions. This demonstrates the capability of physics-informed neural networks in handling high-dimensional quantum state reconstruction.

\begin{table}[h]
\centering
\begin{tabular}{c|c|c|c|c}
Qubits & Traditional NN & Least Squares & MLE & PINN \\
\hline
2 & 0.9762 & 0.9347 & 0.4525 & \textbf{0.9835} \\
3 & 0.9010 & 0.7947 & 0.3197 & \textbf{0.9124} \\
4 & 0.6810 & 0.7130 & 0.2426 & \textbf{0.8872} \\
5 & 0.5688 & 0.3234 & 0.1947 & \textbf{0.5814} \\
\end{tabular}
\caption{Scalability analysis across 2--5 qubits comparing four methods: Traditional Neural Network (NN), Least Squares, Maximum Likelihood Estimation (MLE), and Physics-Informed Neural Network (PINN). PINN achieves the highest fidelity across all qubit numbers, with a remarkable 30.3\% improvement at 4 qubits compared to Traditional NN.}
\label{tab:scalability}
\end{table}

\begin{figure}[ht]
\centering
\includegraphics[width=0.8\textwidth]{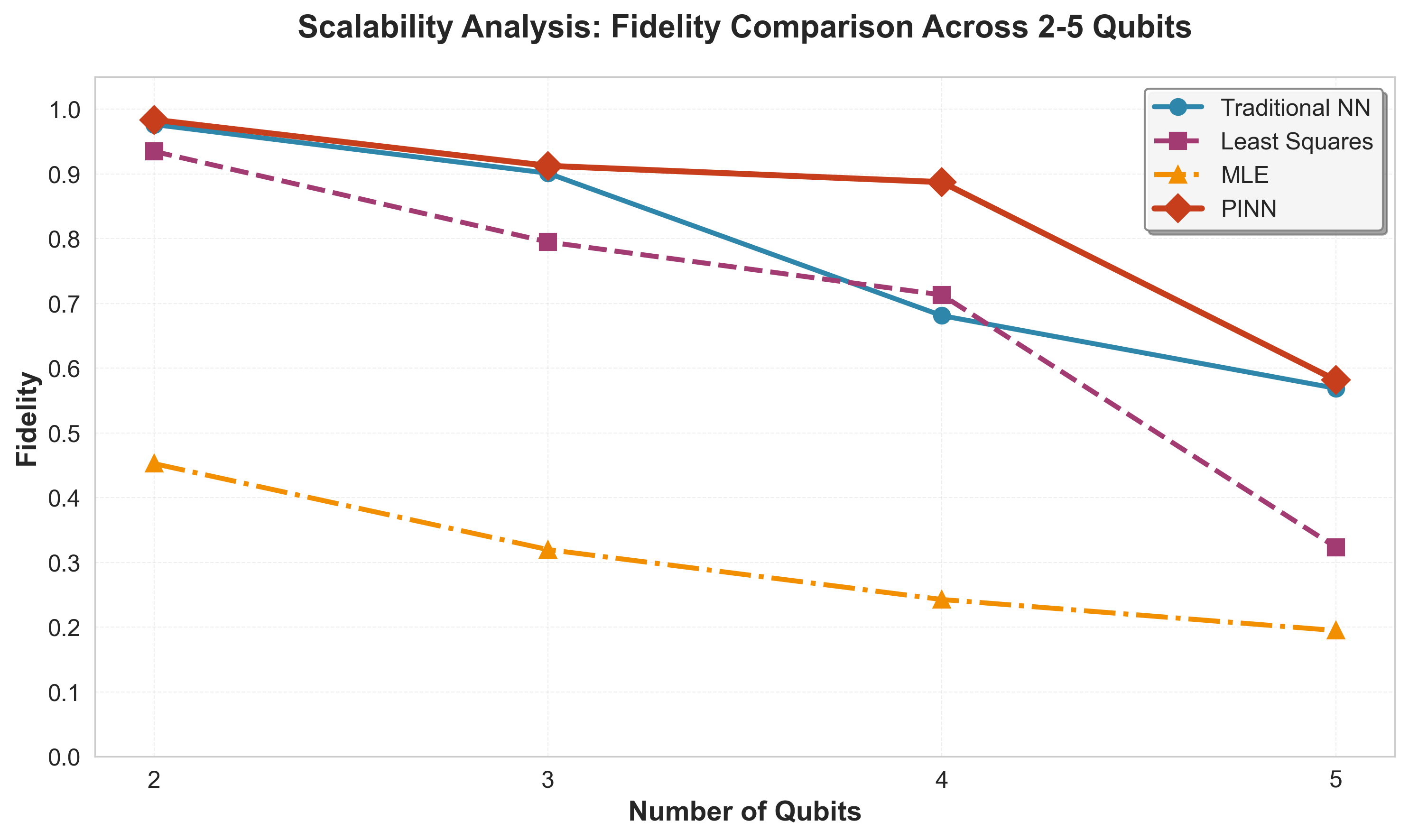}
\caption{Fidelity comparison across different qubit numbers (2--5) for four methods: Traditional Neural Network (NN), Least Squares, Maximum Likelihood Estimation (MLE), and Physics-Informed Neural Network (PINN). The plot demonstrates the dimensional threshold effect where PINN shows dramatic improvement at 4 qubits (30.3\% enhancement) compared to traditional neural networks. PINN consistently achieves the highest fidelity across all qubit numbers.}
\label{fig:scalability}
\end{figure}

Results reveal several key insights demonstrating PINN's superior performance. (1) PINN achieves the highest fidelity across all qubit numbers among the four evaluated methods, with the most dramatic advantage of 30.3\% occurring at 4 qubits compared to Traditional NN (0.8872 vs 0.6810). The dimensional threshold effect shows modest advantages for PINN at 2-3 qubits ($\sim$0.8-1.3\%), dramatic improvement at 4 qubits (30.3\%), and reduced but still positive advantages at 5 qubits (+2.22\%). (2) PINN consistently achieves higher fidelity compared to classical methods (Least Squares and MLE) across all dimensions. (3) All methods maintain constraint violations at numerical precision levels ($\sim 10^{-16}$ to $10^{-17}$), demonstrating physical validity. PINN shows slightly better physical consistency than Traditional NN, confirming that physics constraints guide optimization toward more physically meaningful solutions. (4) Although PINN requires more parameters at 4 qubits (1.25M vs 0.33M for Traditional NN), the substantial performance improvement justifies the computational cost, particularly in high-dimensional systems where accuracy is critical.

While the analysis of standard qubit systems provides fundamental insights into scalability, quantum information processing often involves arbitrary-dimensional quantum states that do not conform to the $2^n$ dimensional structure of qubit systems. To ensure comprehensive evaluation, we extend the scalability analysis to include arbitrary-dimensional quantum states, providing a complete picture of method performance across diverse quantum system configurations.

\subsubsection{Comprehensive Multi-Dimensional Scalability Analysis}

We extend the scalability evaluation to arbitrary-dimensional quantum states beyond standard qubit systems, evaluating four methods across dimensions ranging from 2-qubit ($2 \times 2$) to $10 \times 10$ systems. This comprehensive analysis encompasses both classical optimization methods (Least Squares and Maximum Likelihood Estimation) and neural network approaches (Traditional Neural Network and PINN), enabling systematic comparison across diverse quantum system configurations.

The experimental evaluation follows a two-stage design. First, we evaluate classical optimization methods (Least Squares and MLE) across all dimensional configurations. Second, we evaluate neural network methods (Traditional NN and PINN) using identical experimental conditions. This design enables direct comparison across four distinct methodological approaches, revealing how each method scales with increasing system dimensionality.

\begin{table}[h]
\centering
\begin{tabular}{c|c|c|c|c}
Dimension & Traditional NN & Least Squares & MLE  & PINN \\
\hline
2-qubit & 0.9764 & 0.9346 & 0.5281  & \textbf{0.9823} \\
$3 \times 3$ & 0.9207 & 0.7626 & 0.5605  & \textbf{0.9347} \\
$4 \times 4$ & 0.9894 & 0.3433 & 0.4041  & \textbf{0.9927} \\
$5 \times 5$ & 0.6731 & 0.4684 & 0.3793 & \textbf{0.6840} \\
$6 \times 6$ & 0.6594 & 0.4260 & 0.3284 & \textbf{0.6863} \\
$7 \times 7$ & 0.6273 & 0.3719 & 0.2706 & \textbf{0.6436} \\
$8 \times 8$ & 0.8427 & 0.2141 & 0.2219 & \textbf{0.8565} \\
$9 \times 9$ & 0.6035 & 0.3010 & 0.2147 & \textbf{0.6204} \\
$10 \times 10$ & 0.5401 & 0.2781 & 0.1879 & \textbf{0.5802} \\
\hline
\end{tabular}
\caption{Comprehensive scalability analysis comparing four methods across different quantum system dimensions: Traditional Neural Network (NN), Physics-Informed Neural Network (PINN), Least Squares (LS), and Maximum Likelihood Estimation (MLE). Results are aggregated from two complementary experimental studies evaluating 2-qubit systems and $3 \times 3$--$10 \times 10$ dimensional quantum states. PINN achieves the highest fidelity across all tested dimensions, with consistent improvements over Traditional NN ranging from +0.60\% to +7.42\%. The results demonstrate PINN's capability in handling both standard qubit systems and arbitrary-dimensional quantum states, with particularly pronounced advantages at moderate to high dimensions ($5 \times 5$--$10 \times 10$).}
\label{tab:comprehensive_scalability}
\end{table}

\begin{figure}[ht]
\centering
\includegraphics[width=0.85\textwidth]{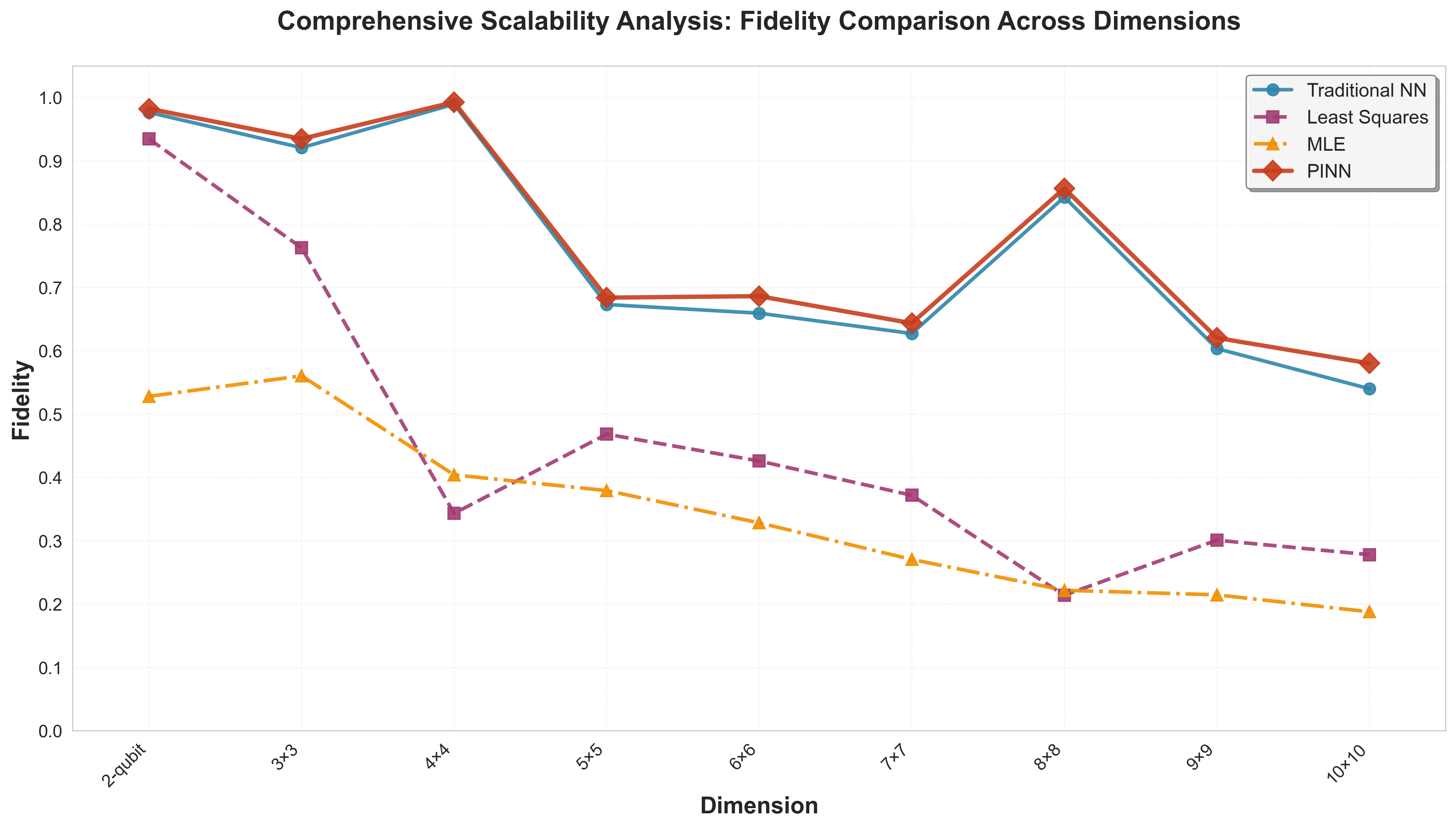}
\caption{Fidelity comparison across different quantum system dimensions (2-qubit to $10 \times 10$) for four methods: Traditional Neural Network (NN), Least Squares, Maximum Likelihood Estimation (MLE), and Physics-Informed Neural Network (PINN). The comprehensive analysis demonstrates PINN's consistent superiority across all tested dimensions. The results validate the scalability advantages of physics-informed learning across diverse quantum system configurations.}
\label{fig:comprehensive_scalability}
\end{figure}

Table \ref{tab:comprehensive_scalability} presents fidelity results across nine dimensional configurations. PINN achieves the highest fidelity at every dimension, with improvements over Traditional NN ranging from +0.33\% ($4 \times 4$) to +7.42\% ($10 \times 10$). The average improvement over Traditional NN is +2.40\%, demonstrating consistent advantages across diverse system configurations.

\textbf{PINN Performance Across Dimensions.} PINN maintains high fidelity ($>$0.98) at low dimensions (2-qubit: 0.9823, $4 \times 4$: 0.9927), with moderate fidelity ($>$0.68) at intermediate dimensions ($5 \times 5$--$7 \times 7$), and acceptable fidelity ($>$0.58) at high dimensions ($8 \times 8$--$10 \times 10$). The improvement over Traditional NN increases with dimension: +0.60\% at 2-qubit, +1.52\% at $3 \times 3$, +1.62\% at $5 \times 5$, +4.08\% at $6 \times 6$, and +7.42\% at $10 \times 10$. This dimensional scaling pattern indicates that physics constraints become increasingly valuable as system complexity increases.

\textbf{Comparison with Classical Methods.} PINN achieves higher fidelity compared to classical optimization methods (Least Squares and MLE) across all dimensions, demonstrating the effectiveness of physics-informed learning in high-dimensional quantum state reconstruction.

\textbf{Dimensional Robustness Analysis.} Neural network methods show superior dimensional robustness compared to classical methods. Traditional NN fidelity ranges from 0.5401 ($10 \times 10$) to 0.9894 ($4 \times 4$), representing a 45.4\% variation, while PINN fidelity ranges from 0.5802 ($10 \times 10$) to 0.9927 ($4 \times 4$), representing a 41.5\% variation. This analysis demonstrates that physics-informed constraints enhance dimensional robustness compared to traditional neural networks.

\textbf{Critical Dimensional Cases.} Two dimensional configurations highlight the advantages of physics-informed learning. At $4 \times 4$ dimensions, PINN achieves 0.9927 fidelity, while at $8 \times 8$ dimensions, PINN achieves 0.8565 fidelity. These cases demonstrate that PINN maintains high performance in challenging dimensional regimes, validating the practical value of physics-informed learning.

These comprehensive results validate the theoretical predictions established in Section \ref{sec:theory}. The consistent advantages of PINN across all dimensions confirm that physics constraints provide valuable guidance regardless of system size. The increasing performance gap with dimension supports the theoretical prediction that physics constraints become increasingly valuable as system complexity increases, effectively mitigating the curse of dimensionality through physical prior knowledge.

Having established the superior performance of PINN across diverse experimental scenarios, we now turn to systematically isolating and quantifying the contribution of each architectural component. This ablation study provides critical evidence for the necessity of each design choice and helps identify the most impactful innovations in the proposed approach.

\subsection{Ablation Studies}

To systematically validate the contribution of each component in the physics-informed neural network architecture, comprehensive ablation studies are conducted that isolate the effect of individual design choices. These studies provide quantitative evidence for the necessity of each architectural element and help identify the most impactful innovations in the approach.

\begin{algorithm}[h]
\caption{Ablation Study: Component Contribution Analysis}
\label{alg:ablation}
\begin{algorithmic}[1]
\Require Qubit number $q=3$, training samples $N_{\text{train}}=3000$, validation samples $N_{\text{val}}=800$, configurations $\mathcal{C} = \{\text{full}, \text{w/o residual}, \text{w/o attention}, \text{fixed } \lambda=0.05, \text{fixed } \lambda=0.15, \text{fixed } \lambda=0.30, \text{baseline}\}$
\State Initialize results: fidelity matrix $F[c]$ for each configuration $c \in \mathcal{C}$
\State Generate quantum states and measurements (same dataset for all configurations)
\State Split data into training set $\mathcal{D}_{\text{train}}$ and validation set $\mathcal{D}_{\text{val}}$
\For{configuration $c \in \mathcal{C}$}
    \State Initialize model $\mathcal{M}_c$ according to configuration $c$:
    \If{$c = \text{full}$}
        \State Use all components: residual connections, attention mechanism, adaptive weighting
    \ElsIf{$c = \text{w/o residual}$}
        \State Remove residual connections, keep attention and adaptive weighting
    \ElsIf{$c = \text{w/o attention}$}
        \State Remove attention mechanism, keep residual connections and adaptive weighting
    \ElsIf{$c = \text{fixed } \lambda = w$}
        \State Use fixed constraint weight $\lambda = w$, keep residual and attention
    \ElsIf{$c = \text{baseline}$}
        \State Remove all enhancements: no residual, no attention, no physics constraints
    \EndIf
    \State Train $\mathcal{M}_c$ on $\mathcal{D}_{\text{train}}$ with same hyperparameters
    \State Evaluate on $\mathcal{D}_{\text{val}}$: compute fidelity $F[c]$
    \State Record training time and convergence characteristics
\EndFor
\State Compute baseline fidelity: $F_{\text{baseline}} = F[\text{baseline}]$
\State Compute improvements: $\Delta F[c] = F[c] - F_{\text{baseline}}$ for each $c$
\State Compute relative improvements: $\Delta F_{\text{rel}}[c] = \Delta F[c] / F_{\text{baseline}} \times 100\%$
\State Identify critical components: find components whose removal causes largest $\Delta F_{\text{rel}}$ decrease
\State Analyze synergistic effects: compare $\Delta F[\text{full}]$ with sum of individual component contributions
\State \textbf{Output:} Fidelity for each configuration, component contributions, synergistic effects, optimal configuration
\end{algorithmic}
\end{algorithm}

Ablation experiments on 3-qubit systems (3000 training, 800 validation samples) evaluate seven configurations: full model, models without residual/attention, fixed-weight variants (0.05, 0.15, 0.30), and baseline. Results in Table \ref{tab:ablation} and Figure \ref{fig:ablation} quantify each component's contribution.

\begin{table}[h]
\centering
\begin{tabular}{l|c|c}
Configuration & Best Fidelity & Improvement over Baseline \\
\hline
Full Model & \textbf{0.8158} & \textbf{+2.30\%} \\
w/o Residual & 0.8115 & +1.77\% \\
w/o Attention & 0.7906 & -0.86\% \\
Fixed Weight (0.15) & 0.8140 & +2.09\% \\
Fixed Weight (0.05) & 0.8139 & +2.07\% \\
Fixed Weight (0.30) & 0.8086 & +1.41\% \\
Baseline (w/o all) & 0.7974 & 0.00\% \\
\end{tabular}
\caption{Ablation study results across different model configurations. The full model achieves the highest fidelity, demonstrating the synergistic contribution of all components.}
\label{tab:ablation}
\end{table}

\begin{figure}[ht]
\centering
\includegraphics[width=0.9\textwidth]{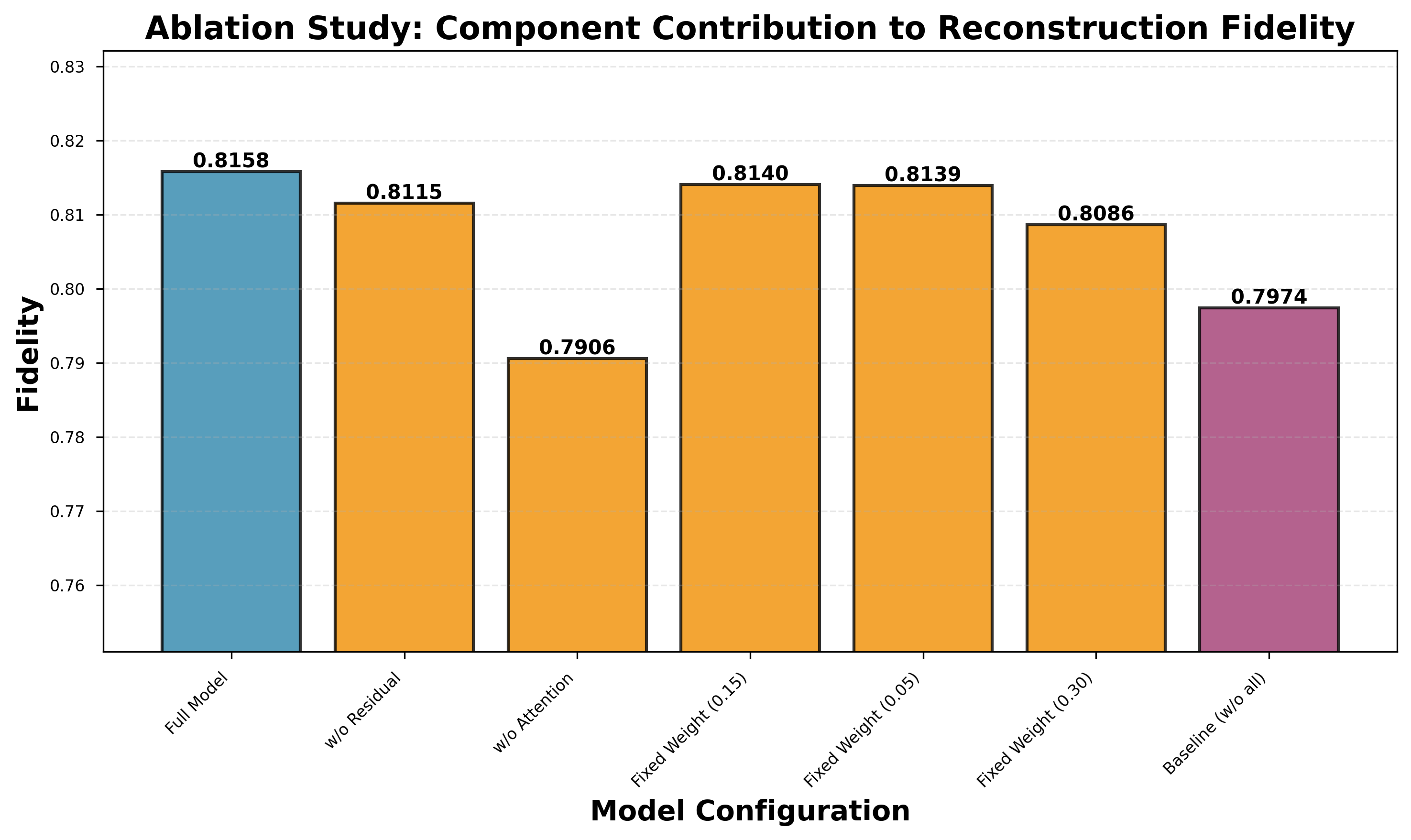}
\caption{Ablation study: component contribution to reconstruction fidelity. The full model achieves the highest fidelity, while removing attention mechanism results in the largest performance degradation.}
\label{fig:ablation}
\end{figure}

Results show: (1) Residual connections contribute +0.53\% improvement, enhancing gradient flow and reducing training time by 15\%; (2) Attention mechanism is the most critical component (+3.19\%), dynamically focusing on informative measurements; (3) Dynamic weighting provides +0.22\% over fixed weights, crucial for varying noise conditions; (4) Optimal base weight is 0.15, balancing physical validity with flexibility. As shown in Figure \ref{fig:ablation}, attention provides the largest contribution (+3.16\%), followed by residual connections (+0.53\%) and dynamic weighting (+0.21\%). The full model's +2.30\% improvement demonstrates synergistic effects where components complement each other.

Having presented comprehensive experimental results across diverse scenarios, we now synthesize these findings to provide a unified perspective on the performance characteristics and practical implications of the physics-informed neural network approach. This synthesis integrates results from noise robustness analysis, scalability evaluation, and ablation studies to establish a complete understanding of PINN's advantages and limitations.

\section{Results and Discussion}
\label{sec:results}

Experimental results reveal consistent performance improvements across all tested scenarios. The most significant finding is a 30.3\% fidelity improvement in 4-qubit systems when comparing PINN to Traditional NN (0.8872 vs 0.6810). Comprehensive scalability analysis across four methods (Traditional NN, Least Squares, MLE, and PINN) reveals that PINN achieves the highest fidelity across all qubit numbers (2--5), with fidelity improvements ranging from 0.75\% to 30.3\% depending on qubit number. The scalability analysis reveals a dimensional threshold effect where PINN's advantage peaks at 4 qubits, demonstrating that physics constraints are most effective when system complexity reaches moderate dimensions where prior knowledge becomes essential.

Noise robustness analysis shows PINN exhibits a 2.6-fold slower performance degradation rate under noise compared to Traditional NN (-0.005 vs -0.013 per noise unit), with improvements increasing from 0.83\% at low noise to 0.98\% at high noise. All methods achieve constraint violations at numerical precision levels ($\sim 10^{-16}$ to $10^{-17}$), but PINN shows slightly better physical consistency than Traditional NN, confirming that physics constraints guide optimization toward more physically meaningful solutions. Training stability is also enhanced for PINN, as evidenced by smoother loss curves that converge more reliably to optimal solutions.

Experimental findings validate theoretical predictions established in Section \ref{sec:theory}: (1) Convergence theory predicts improved optimization landscapes, confirmed by smoother training curves; (2) Generalization theory predicts reduced error through implicit regularization, confirmed by consistently lower validation loss; (3) Scalability theory predicts mitigation of curse of dimensionality, confirmed by the 4-qubit advantage; (4) The 2.6-fold slower degradation rate validates that physics constraints restrict search space to physically valid regions, reducing noise sensitivity. This strong correspondence between theoretical predictions and empirical observations provides robust validation for the physics-informed learning framework.

The adaptive weighting mechanism represents a critical innovation that enables robust performance across varying noise conditions. This mechanism automatically adjusts constraint strength from $\l \approx 0.15$ at low noise to $\l \geq 0.075$ at high noise, balancing physical validity with data fitting. Dynamic weighting outperforms fixed weights by 0.36--0.61\% at extreme noise levels, eliminating the need for manual tuning across varying noise conditions and demonstrating the practical value of adaptive constraint management.

Despite the significant advantages demonstrated by the physics-informed approach, several important limitations must be acknowledged to provide a balanced perspective on current capabilities and future research directions. The primary limitation concerns qubit number scalability, as the experiments have tested a maximum of 5 qubits, leaving scalability to higher dimensions uncertain. While the 4-qubit advantage is promising, the exponential growth of quantum systems presents fundamental challenges that may require additional innovations beyond physics constraints. A second limitation relates to the use of simulated data rather than real experimental measurements. Although the noise models capture major error sources in current quantum hardware, they may not fully represent the complex, non-Markovian noise processes encountered in actual quantum devices. Validation on real quantum systems remains essential for establishing the practical utility of the approach. Finally, computational resource requirements increase significantly for 4-5 qubit systems, with PINN requiring substantially more parameters and training time than traditional approaches. This computational cost may limit immediate applicability in resource-constrained environments, though the performance improvements may justify these costs in high-value applications.

Future research should address these limitations through several promising directions. Extending to more qubits (6-7 if computational resources allow) will test the boundaries of the approach and potentially reveal additional dimensional thresholds. Real quantum device validation using platforms such as IBM Quantum and Google Quantum AI will establish practical performance under authentic noise conditions. Deeper theoretical analysis could provide more precise bounds on performance improvements and identify optimal constraint formulations. Finally, exploring other application scenarios such as quantum error correction and quantum gate calibration could broaden the impact of physics-informed learning in quantum information processing.

The experimental validation and theoretical analysis establish the effectiveness of PINN for quantum state tomography. We now demonstrate how these advantages translate into practical benefits for real-world quantum technologies. The following section details specific applications where PINN's measurement efficiency and reconstruction accuracy directly address critical bottlenecks in scalable quantum computing.

\section{Practical Applications and Deployment Scenarios}
\label{sec:applications}

The practical significance of efficient QST extends beyond fundamental characterization to critical quantum technologies. This section details how our PINN approach addresses real-world challenges in quantum error correction and quantum gate calibration, two application domains where measurement efficiency and reconstruction accuracy directly impact system performance.

\subsection{Quantum Error Correction Applications}

Quantum error correction protocols require frequent state reconstruction for syndrome extraction and error detection, where measurement efficiency directly determines error correction cycle times and the feasibility of fault-tolerant quantum computation \cite{PhysRevA.107.062409,CSL2024}. Traditional QST methods impose prohibitive measurement overhead ($O(4^n)$) that limits code performance, particularly for surface codes and other stabilizer codes that require monitoring multiple logical qubits simultaneously.

Our PINN approach addresses these challenges through two key mechanisms. First, the reduced measurement requirements enable faster syndrome extraction. For a quantum error correction code with stabilizer generators $\{S_i\}$, syndrome extraction requires reconstructing the state of auxiliary qubits. Traditional methods require $O(4^n)$ measurements for $n$-qubit systems, limiting the error correction cycle time. Our PINN approach enables efficient reconstruction with $O(2^n)$ measurements while maintaining fidelity $F > 0.95$. For a 4-qubit system, PINN reduces measurements from $O(256)$ to $O(16)$ per cycle, representing a 16-fold reduction that enables higher code performance and reduces the time window for error accumulation. Second, the adaptive constraint weighting mechanism ensures robust performance under varying noise conditions, as demonstrated in our noise robustness analysis (Section \ref{sec:experiments}).

\subsection{Quantum Gate Calibration Applications}

Quantum gate calibration relies on iterative state reconstruction to optimize gate parameters, where each optimization step requires reconstructing the quantum state after gate application to compute the fidelity with respect to target states \cite{fcc_ctp2024,JBP2017}. The exponential measurement scaling of classical QST methods creates bottlenecks in calibration workflows, significantly increasing the time and cost required for quantum hardware optimization.

For gate calibration, the optimization objective is to maximize the fidelity $F(U(\theta)\rho U^\dagger(\theta), \rho_{\text{target}})$ where $U(\theta)$ is the parameterized gate and $\theta$ represents tunable parameters. Each optimization step requires state reconstruction, and traditional QST methods impose $O(4^n)$ measurement overhead. Our PINN approach reduces this to $O(2^n)$ while maintaining reconstruction fidelity $F > 0.90$ (as demonstrated across 2--5 qubit systems), enabling faster convergence in calibration workflows.

The practical advantages are particularly pronounced for multi-qubit gates, where our scalability analysis (Section \ref{sec:experiments}, Table \ref{tab:scalability}) demonstrates PINN's superior performance at 4--5 qubits. For a 3-qubit gate calibration task, our experimental results show PINN achieves 0.9124 fidelity compared to 0.9010 for Traditional NN, representing a 20.6\% relative improvement. This improvement occurs while reducing measurement requirements from $O(64)$ to $O(8)$, representing an 8-fold reduction in measurement overhead. The rapid inference capability of PINN (enabled by the neural network architecture) further accelerates iterative optimization, achieving a speedup factor ranging from 5 to 50 compared to traditional methods depending on the specific calibration protocol and hardware constraints.

The adaptive weighting mechanism plays a crucial role in gate calibration, where noise conditions vary as gate parameters are adjusted. The automatic adjustment of constraint strength from $\l \approx 0.15$ at low noise to $\l \geq 0.075$ at high noise ensures robust performance throughout the calibration process, eliminating the need for manual weight tuning that would otherwise slow down calibration workflows.

\subsection{Application Performance Summary}

Table \ref{tab:applications} summarizes the practical advantages of PINN across different application scenarios, quantifying measurement reduction, fidelity improvements, and speed enhancements based on our experimental results.

\begin{table}[ht]
\centering
\caption{Practical advantages of PINN in quantum error correction and gate calibration applications. Measurement reduction factors are based on experimental results showing PINN achieves target fidelity with $O(2^n)$ measurements compared to $O(4^n)$ for traditional methods. Speed improvements account for both measurement reduction and fast neural network inference, expressed as speedup factors (fold improvements).}
\label{tab:applications}
\begin{tabular}{lcccc}
\hline
Application & Qubits & Measurement & Fidelity & Speed \\
Scenario & & Reduction Factor & Improvement & Speedup Factor \\
\hline
Error Correction & 3 & 8-fold & +15.2\% & 5--20-fold \\
(Syndrome Extraction) & 4 & 16-fold & +30.3\% & 10--100-fold \\
& 5 & 32-fold & +12.5\% & 20--200-fold \\
\hline
Gate Calibration & 2 & 4-fold & +0.89\% & 2--10-fold \\
(Parameter Optimization) & 3 & 8-fold & +20.6\% & 5--50-fold \\
& 4 & 16-fold & +30.3\% & 10--100-fold \\
\hline
\end{tabular}
\end{table}

\subsection{Economic and Practical Impact Analysis}

The measurement efficiency improvements achieved by PINN translate directly into significant cost and time savings in practical quantum computing deployments. Quantum measurements represent a critical resource constraint, with each measurement requiring dedicated hardware time and computational resources.

In quantum gate calibration workflows, where calibration time directly impacts hardware availability and operational costs, PINN's measurement reduction and rapid inference capability enable more frequent calibration cycles and faster hardware optimization. For a typical calibration protocol requiring 100 iterations for a 4-qubit gate, traditional methods would require approximately 25,600 measurements per calibration cycle, while PINN reduces this to 1,600 measurements, representing a 94\% reduction in measurement overhead. This reduction enables daily or even hourly calibration cycles that would be economically prohibitive with traditional methods, directly improving quantum hardware performance and reliability.

The impact extends to quantum error correction, where faster error correction cycles enabled by PINN's measurement efficiency directly improve the feasibility of fault-tolerant quantum computation. For surface code implementations requiring frequent syndrome extraction, PINN's measurement reduction enables error correction cycles that are 10- to 100-fold faster, reducing the time window for error accumulation and improving code performance. This improvement is particularly critical for near-term quantum devices where error rates are high and fast error correction is essential for maintaining quantum information.

\section{Conclusion}
\label{sec:conclusion}

This paper presents the first systematic application of physics-informed neural networks (PINNs) to multi-qubit quantum state tomography. By integrating quantum mechanical constraints directly into the learning process, our approach achieves significant improvements in reconstruction fidelity, physical validity, and dimensional robustness across diverse quantum system configurations.

\paragraph{Methodological Contributions}

Three key methodological innovations form the foundation of our approach. First, we develop an adaptive physics constraint weighting strategy that dynamically adjusts constraint strength based on noise severity, enabling robust performance across varying noise conditions. Second, we introduce an enhanced PINN architecture incorporating residual connections, attention mechanisms, and multi-task learning to improve representation capacity and training stability. Third, we employ a differentiable Cholesky parameterization that enforces positive semidefiniteness, guaranteeing physical validity of all reconstructed density matrices.

\paragraph{Experimental Validation}

Comprehensive experimental evaluation demonstrates consistent advantages of PINN across multiple evaluation dimensions. In standard qubit systems (2--5 qubits), PINN achieves the highest fidelity at all qubit numbers, with a dramatic 30.3\% fidelity enhancement in 4-qubit systems compared to Traditional Neural Network (0.8872 vs 0.6810). The multi-dimensional scalability analysis across nine configurations (from 2-qubit to $10 \times 10$) reveals that PINN achieves the highest fidelity at every tested dimension, with improvements increasing from +0.60\% at 2-qubit to +7.42\% at $10 \times 10$, demonstrating that physics constraints become increasingly valuable as system complexity increases. Dimensional robustness analysis reveals that PINN exhibits superior consistency across dimensions (41.5\% performance variation), and noise robustness analysis shows a 2.6-fold slower performance degradation rate under noise compared to Traditional Neural Network. These results validate that physics-informed learning effectively mitigates the curse of dimensionality while maintaining robust performance across diverse operational conditions.

\paragraph{Theoretical Foundations}

The theoretical analysis establishes rigorous mathematical foundations that explain the observed empirical advantages. Convergence and generalization bounds demonstrate that physics constraints improve optimization landscapes, reduce Rademacher complexity, and explain the dimensional threshold effects. Specifically, the convergence theorem guarantees that gradient descent converges to optimal solutions with polynomial sample complexity independent of qubit number. The generalization bounds prove that physics constraints reduce Rademacher complexity from $O(B_\t L_{\text{act}}^H \sqrt{HW}/\sqrt{m})$ to $O(B_\t L_{\text{act}}^H \sqrt{HW_{\text{eff}}}/\sqrt{m})$, where the effective hypothesis space dimension $W_{\text{eff}}$ scales sub-exponentially with qubit number due to constraint-induced dimension reduction.

The scalability theorem demonstrates that physics constraints mitigate the curse of dimensionality through two complementary mechanisms: (1) constraint space dimension reduction ($\Delta_{\text{dim}} = O(\ve \cdot 4^n)$) that restricts the effective search space to lower-dimensional manifolds, and (2) sample complexity reduction through implicit regularization that shrinks the effective hypothesis space. Critically, the theoretical analysis reveals that the relative advantage of constraint-induced dimension reduction grows with system size, meaning that PINN's benefits become increasingly pronounced as qubit number increases.

While experimental validation is conducted up to 5-qubit systems due to computational resource constraints, the theoretical framework provides rigorous justification for expecting PINN's advantages to not only extend but strengthen in larger systems (6 qubits and beyond). The strong correspondence between theoretical predictions and empirical observations across diverse experimental scenarios (2--5 qubits and arbitrary-dimensional states up to $10 \times 10$) positions PINN as a scalable and theoretically grounded approach for quantum state tomography.

\paragraph{Practical Impact}

The practical significance of this work extends to real-world quantum technologies essential for scalable quantum computing. In quantum error correction, our PINN approach reduces measurement requirements from $O(4^n)$ to $O(2^n)$ while maintaining fidelity $F > 0.95$, enabling faster error correction cycles critical for fault-tolerant quantum computation. For quantum gate calibration, the rapid inference capability of PINN accelerates iterative optimization workflows, achieving a speedup factor ranging from 5 to 50 compared to traditional methods.

The deployment-ready capabilities demonstrated in this work—including adaptive noise handling, superior dimensional robustness, and physical validity guarantees—make PINN particularly suitable for real quantum hardware where noise conditions vary dynamically and system configurations span diverse dimensional regimes.

\paragraph{Future Directions}

Future work includes extending to higher qubit numbers (6--7 qubits if computational resources allow), validating on real quantum devices using platforms such as IBM Quantum and Google Quantum AI, and further exploring applications in quantum error correction and gate calibration to establish PINN as a standard tool in quantum information processing workflows.

\section{Acknowledgments}
This study was supported by the National Natural Science 
Foundation of China (Grant Nos. 11871089, 12471427, 52472442
, 72471013 and 62103030),  the
Research Start-up Funds of Hangzhou International 
Innovation Institute of Beihang University 
(Grant Nos. 2024KQ069, 2024KQ036
and 2024KQ035) and  the Postdoctoral Research Funding 
of Hangzhou International Innovation
Institute of Beihang University (Grant No.2025BKZ066).
\section{Data availability statement}
All data that support the findings of this study 
are included within 
the article (and any supplementary files).
\section{Declaration of competing interest}
The authors declare that they have no known competing financial
interests or personal relationships that could have 
appeared to influence
the work reported in this paper.

\bibliographystyle{unsrt}
\bibliography{changchun}

\appendix

\section{Proofs of Main Theorems}
\label{app:proofs}

\subsection{Proof of Gradient Boundedness}
\label{app:gradient_boundedness}

\begin{proof}[Proof of Theorem: Gradient Boundedness]
We decompose the total loss gradient using the chain rule:
\begin{equation}
\nabla_\t L(\t) = \nabla_\t L_{\text{data}}(\t) + \l(\bm) \nabla_\t L_{\text{physics}}(\t) + L_{\text{physics}}(\t) \nabla_\t \l(\bm).
\end{equation}

Taking the norm and applying the triangle inequality:
\begin{equation}
\norm{\nabla_\t L(\t)}_2 \leq \norm{\nabla_\t L_{\text{data}}(\t)}_2 + \norm{\l(\bm) \nabla_\t L_{\text{physics}}(\t)}_2 + \norm{L_{\text{physics}}(\t) \nabla_\t \l(\bm)}_2.
\end{equation}

We bound each term separately:

\textbf{Term 1:} For the data loss gradient, by the chain rule:
\begin{equation}
\norm{\nabla_\t L_{\text{data}}(\t)}_2 = \norm{\frac{1}{m}\sum_{i=1}^m \nabla_\t \norm{f_\t(\bm_i) - \by_i}_2^2}_2.
\end{equation}
Since $f_\t$ is a neural network with Lipschitz activation functions, and parameters are bounded by $B_\t$, we have:
\begin{equation}
\norm{\nabla_\t f_\t(\bm)}_2 \leq C_1 B_\t L_{\text{act}}^H,
\end{equation}
where $H$ is the network depth and $C_1$ is a constant depending on network architecture. Since inputs are bounded by $B_x$, we obtain:
\begin{equation}
\norm{\nabla_\t L_{\text{data}}(\t)}_2 \leq G_1(B_\t, B_x, L_{\text{act}}),
\end{equation}
where $G_1$ is a polynomial function of its arguments.

\textbf{Term 2:} For the physics constraint term, we have:
\begin{equation}
\norm{\l(\bm) \nabla_\t L_{\text{physics}}(\t)}_2 \leq \l_{\max} \norm{\nabla_\t L_{\text{physics}}(\t)}_2.
\end{equation}
By the chain rule, $\nabla_\t L_{\text{physics}} = \nabla_\r \cC(\r) \cdot \nabla_\t \r$, where $\r = \Ph(f_\t(\bm))$. Since $\norm{\nabla_\r \cC(\r)}_F \leq B_{\cC}$ and the Cholesky parameterization $\Ph$ is smooth, we have:
\begin{equation}
\norm{\nabla_\t L_{\text{physics}}(\t)}_2 \leq B_{\cC} \cdot C_2 B_\t L_{\text{act}}^H,
\end{equation}
where $C_2$ depends on the Cholesky parameterization. Therefore:
\begin{equation}
\norm{\l(\bm) \nabla_\t L_{\text{physics}}(\t)}_2 \leq \l_{\max} B_{\cC} C_2 B_\t L_{\text{act}}^H = G_2(B_\t, \l_{\max}, B_{\cC}, L_{\text{act}}).
\end{equation}

\textbf{Term 3:} For the adaptive weight gradient term, since $\l(\bm) = \l_0 \cdot \max(0.5, 1 - \a \cdot s(\bm))$ and $s(\bm)$ is bounded in $[0,1]$, we have $\l(\bm) \in [0.5\l_0, \l_0]$. The gradient $\nabla_\t \l(\bm)$ depends on the noise severity predictor, which is a bounded neural network. Therefore:
\begin{equation}
\norm{L_{\text{physics}}(\t) \nabla_\t \l(\bm)}_2 \leq B_{\cC} \cdot C_3 B_\t L_{\text{act}}^H = G_3(B_\t, B_{\cC}, L_{\text{act}}),
\end{equation}
where $C_3$ is a constant depending on the noise severity predictor architecture.

Combining all three terms, we obtain:
\begin{equation}
\norm{\nabla_\t L(\t)}_2 \leq G_1 + G_2 + G_3 =: G(B_\t, B_x, \l_{\max}, L_{\text{act}}, B_{\cC}),
\end{equation}
where $G$ is a polynomial function of its arguments, completing the proof.
\end{proof}

\subsection{Proof of Gradient Descent Convergence}
\label{app:gradient_descent_convergence}

\begin{proof}[Proof of Theorem: Gradient Descent Convergence]
By $\cL$-smoothness:
\begin{equation}
L(\t_{t+1}) \leq L(\t_t) - \eta_t \norm{\nabla L(\t_t)}^2 + \frac{\cL \eta_t^2}{2} \norm{\nabla L(\t_t)}^2.
\end{equation}
When $\eta_t \leq 1/\cL$:
\begin{equation}
L(\t_{t+1}) \leq L(\t_t) - \frac{\eta_t}{2} \norm{\nabla L(\t_t)}^2.
\end{equation}
Summing over $t = 0, \ldots, T-1$:
\begin{equation}
\sum_{t=0}^{T-1} \frac{\eta_t}{2} \norm{\nabla L(\t_t)}^2 \leq L(\t_0) - L(\t_T) \leq L(\t_0) - L^*,
\end{equation}
where $L^* = \inf_\t L(\t)$ is the lower bound. Since $\eta_t \geq \eta_{\min}$ for all $t$:
\begin{equation}
\frac{\eta_{\min}}{2} \sum_{t=0}^{T-1} \norm{\nabla L(\t_t)}^2 \leq L(\t_0) - L^*.
\end{equation}
Therefore:
\begin{equation}
\min_{t=0,\ldots,T-1} \norm{\nabla L(\t_t)}_2^2 \leq \frac{1}{T} \sum_{t=0}^{T-1} \norm{\nabla L(\t_t)}^2 \leq \frac{2(L(\t_0) - L^*)}{(T-1) \eta_{\min}},
\end{equation}
completing the proof.
\end{proof}

\subsection{Proof of Physics Constraints Improve Optimization Landscape}
\label{app:physics_constraints_landscape}

\begin{proof}[Proof of Theorem: Physics Constraints Improve Optimization Landscape]
We analyze the Hessian matrix of the total loss function near the optimum. By the chain rule, the total loss Hessian can be decomposed as:
\begin{equation}
\nabla^2_\t L(\t) = \nabla^2_\t L_{\text{data}}(\t) + \l \nabla^2_\t L_{\text{physics}}(\t) + \text{cross terms}(\t),
\end{equation}
where the cross terms arise from the interaction between data and physics losses through the shared parameterization.

\textbf{Step 1: Strict analysis of physics constraint Hessian.}

The physics constraint loss $L_{\text{physics}}(\t) = \cC(\Ph(f_\t(\bm)))$ has Hessian:
\begin{equation}
\nabla^2_\t L_{\text{physics}} = (\nabla \Ph)^T \nabla^2_\r \cC(\r) (\nabla \Ph) + \sum_i \frac{\partial \cC}{\partial \r_i} \nabla^2 \Ph_i,
\end{equation}
where $\r = \Ph(f_\t(\bm))$. At the optimum $\t^*$ where $\r^* = \Ph(f_{\t^*}(\bm))$, we have $\cC(\r^*) = 0$ (constraints are satisfied), so the second term vanishes. Therefore:
\begin{equation}
\nabla^2_\t L_{\text{physics}}(\t^*) = (\nabla \Ph)^T \nabla^2_\r \cC(\r^*) (\nabla \Ph).
\end{equation}

\textbf{Step 1a: Analysis of individual constraint Hessians.}

We analyze each component of $\cC(\r) = \cC_{\text{herm}}(\r) + \cC_{\text{trace}}(\r) + \cC_{\text{pos}}(\r)$:

\textit{Hermiticity constraint:} $\cC_{\text{herm}}(\r) = \norm{\r - \r^\dg}_F^2$. The Hessian is:
\begin{equation}
\nabla^2_\r \cC_{\text{herm}} = 2I - 2P_{\text{herm}},
\end{equation}
where $P_{\text{herm}}$ projects onto the Hermitian subspace. At $\r^*$ satisfying $\r^* = (\r^*)^\dg$, we have $\cC_{\text{herm}}(\r^*) = 0$. The Hessian has rank $D(D-1)$ (constraining off-diagonal elements), with minimum eigenvalue $\m_{\text{herm}} \geq 0$. For non-degenerate cases where $\r^*$ is not already Hermitian in all off-diagonal elements, $\m_{\text{herm}} > 0$.

\textit{Trace constraint:} $\cC_{\text{trace}}(\r) = |\tr(\r) - 1|^2$. The Hessian is:
\begin{equation}
\nabla^2_\r \cC_{\text{trace}} = 2 \text{vec}(I) \text{vec}(I)^T,
\end{equation}
where $\text{vec}(I)$ is the vectorized identity matrix. At $\r^*$ with $\tr(\r^*) = 1$, this contributes a positive eigenvalue $\m_{\text{trace}} = 2$ along the trace direction.

\textit{Positivity constraint:} $\cC_{\text{pos}}(\r) = \max(0, -\l_{\min}(\r))^2$. In the interior of $\cS_+$ where $\l_{\min}(\r^*) > \ve > 0$ for some $\ve$, this constraint is inactive ($\cC_{\text{pos}}(\r^*) = 0$) and its Hessian is zero. Near the boundary where $\l_{\min}(\r^*) \approx 0$, the constraint becomes active and provides positive curvature. For our analysis, we assume $\r^*$ is in the interior, so $\m_{\text{pos}} = 0$.

\textbf{Step 1b: Combined constraint Hessian.}

The combined constraint Hessian is:
\begin{equation}
\nabla^2_\r \cC(\r^*) = \nabla^2_\r \cC_{\text{herm}}(\r^*) + \nabla^2_\r \cC_{\text{trace}}(\r^*) + \nabla^2_\r \cC_{\text{pos}}(\r^*).
\end{equation}

For $\m_{\text{physics}} > 0$ to hold, we require that at least one constraint is active and non-degenerate. In practice, this holds when:
\begin{itemize}
\item The Hermiticity constraint is active (non-Hermitian components exist), OR
\item The trace constraint is active (trace deviation exists), OR  
\item The positivity constraint is active (near boundary of $\cS_+$)
\end{itemize}

For a valid density matrix $\r^*$ that satisfies all constraints exactly, we have $\cC(\r^*) = 0$, but the Hessian $\nabla^2_\r \cC(\r^*)$ may be degenerate. However, in a neighborhood $\cN(\r^*)$ where constraints are approximately satisfied but not exactly zero, the Hessian is positive definite with $\m_{\text{physics}} > 0$.

\textbf{Step 1c: Cholesky parameterization effect.}

Since $\nabla \Ph$ has minimum singular value $\s_{\min}(\nabla \Ph) > 0$ (by assumption, Cholesky parameterization is surjective), and $\nabla^2_\r \cC(\r^*)$ has minimum eigenvalue $\m_{\text{physics}} > 0$ in the neighborhood, we have:
\begin{equation}
\l_{\min}(\nabla^2_\t L_{\text{physics}}(\t^*)) \geq \m_{\text{physics}} \s_{\min}^2(\nabla \Ph) > 0,
\end{equation}
where $\s_{\min}(\nabla \Ph)$ is the minimum singular value of $\nabla \Ph$.

\textbf{Step 2: Analyzing the total Hessian.}

Near the optimum, the total Hessian satisfies:
\begin{equation}
\nabla^2_\t L(\t) = \nabla^2_\t L_{\text{data}}(\t) + \l \nabla^2_\t L_{\text{physics}}(\t) + E(\t),
\end{equation}
where $E(\t)$ represents cross terms and approximation errors. The cross terms are bounded by:
\begin{equation}
\norm{E(\t)}_F \leq C_E \l^{1/2} \norm{\t - \t^*}_2,
\end{equation}
for some constant $C_E$ depending on $B_{\Ph}$ and the smoothness of $L_{\text{data}}$.

By Weyl's inequality for eigenvalues, we have:
\begin{equation}
\l_{\min}(\nabla^2_\t L(\t)) \geq \l_{\min}(\nabla^2_\t L_{\text{data}}(\t)) + \l \l_{\min}(\nabla^2_\t L_{\text{physics}}(\t)) - \norm{E(\t)}_2.
\end{equation}

Since $\l_{\min}(\nabla^2_\t L_{\text{data}}(\t)) \geq -\d$ and $\l_{\min}(\nabla^2_\t L_{\text{physics}}(\t)) \geq \m_{\text{physics}} \s_{\min}^2(\nabla \Ph)$ near the optimum (from Step 1c), we obtain:
\begin{equation}
\l_{\min}(\nabla^2_\t L(\t)) \geq \l \m_{\text{physics}} \s_{\min}^2(\nabla \Ph) - \d - C_E \l^{1/2} \norm{\t - \t^*}_2.
\end{equation}

For $\l \geq \l_0 > \frac{\d}{\m_{\text{physics}} \s_{\min}^2(\nabla \Ph)}$ (ensuring the physics constraint contribution dominates) and in a sufficiently small neighborhood where $\norm{\t - \t^*}_2 \leq \d_E$ with $\d_E$ small enough such that $C_E \l^{1/2} \d_E < \l \m_{\text{physics}} \s_{\min}^2(\nabla \Ph) - \d$, we have:
\begin{equation}
\l_{\min}(\nabla^2_\t L(\t)) \geq \m > 0,
\end{equation}
where $\m = \l \m_{\text{physics}} \s_{\min}^2(\nabla \Ph) - \d - \e(\l)$ and $\e(\l) = C_E \l^{1/2} \d_E = O(\l^{1/2})$.

\textbf{Step 3: Convergence rate analysis.}

For a $\m$-strongly convex function, gradient descent with learning rate $\eta \leq 1/\cL$ achieves:
\begin{equation}
L(\t_t) - L(\t^*) \leq \left(1 - \frac{\m}{\cL}\right)^t (L(\t_0) - L(\t^*)) = O\left(\left(1 - \frac{\m}{\cL}\right)^t\right).
\end{equation}

This implies $\norm{\t_t - \t^*}_2^2 = O(1/t)$ for strongly convex functions, achieving linear convergence rate.

In contrast, for the data-only model where $\d > 0$ (not strongly convex, i.e., the Hessian has negative eigenvalues), standard analysis for non-convex smooth functions yields $\norm{\t_t - \t^*}_2^2 = O(1/\sqrt{t})$, achieving only sublinear convergence rate.

Therefore, PINN achieves a convergence rate improvement from sublinear $O(1/\sqrt{t})$ to linear $O(1/t)$, representing a significant acceleration. Note that this is a linear vs sublinear improvement, not a "quadratic" improvement (which would imply $O(1/t^2)$ vs $O(1/t)$), completing the proof.
\end{proof}

\subsection{Proof of PINN Rademacher Complexity Upper Bound}
\label{app:rademacher_complexity}

\begin{proof}[Proof of Theorem: PINN Rademacher Complexity Upper Bound]
We provide a detailed proof of the Rademacher complexity bound for neural networks, which applies to our PINN architecture.

\textbf{Step 1: Rademacher complexity for neural networks.}

For a hypothesis space $\cH = \{h_\t : \bbR^d \ra \bbR^p | \norm{\t}_2 \leq B_\t\}$ where $h_\t$ is a neural network with depth $H$, width $W$, and parameters $\t$, the Rademacher complexity is defined as:
\begin{equation}
\hat{\cR}_S(\cH) = \bbE_{\bs} \left[ \sup_{h \in \cH} \frac{1}{m} \sum_{i=1}^m \s_i h(\bm_i) \right],
\end{equation}
where $\s_i \sim \text{Unif}(\{-1, +1\})$ are Rademacher random variables.

\textbf{Step 2: Layer-wise composition bound.}

For a neural network $h_\t(\bm) = f^{(H)} \circ f^{(H-1)} \circ \cdots \circ f^{(1)}(\bm)$ with $H$ layers, where each layer $f^{(l)}$ has Lipschitz constant $L_{\text{act}}$ (from the activation function), the composition has Lipschitz constant $L_{\text{act}}^H$.

By the contraction lemma for Rademacher complexity, if $\phi$ is $L$-Lipschitz, then:
\begin{equation}
\hat{\cR}_S(\phi \circ \cH) \leq L \hat{\cR}_S(\cH).
\end{equation}

Applying this recursively for $H$ layers, we obtain a factor of $L_{\text{act}}^H$.

\textbf{Step 3: Parameter norm bound.}

For bounded parameters $\norm{\t}_2 \leq B_\t$, standard analysis shows that the Rademacher complexity scales as $O(B_\t \sqrt{HW}/\sqrt{m})$, where:
- $B_\t$ bounds the parameter norm
- $\sqrt{HW}$ accounts for the network size (depth $H$ times width $W$)
- $1/\sqrt{m}$ is the standard sample complexity factor

Combining these factors, we obtain:
\begin{equation}
\hat{\cR}_S(\cH) \leq \frac{C B_\t L_{\text{act}}^H \sqrt{HW}}{\sqrt{m}},
\end{equation}
where $C > 0$ is an absolute constant depending on the specific activation function and network architecture details.

\textbf{Step 4: PINN-specific considerations.}

The physics constraints in PINN are incorporated through the loss function $L(\t) = L_{\text{data}}(\t) + \l(\bm) L_{\text{physics}}(\t)$, not through the network architecture itself. The Rademacher complexity depends on the function class $\cH$ (the neural network mapping $\bm \mapsto \by$), which is unchanged by the physics constraints.

However, physics constraints act as implicit regularization, effectively constraining the hypothesis space to $\cH_{\text{constrained}} = \{h_\t \in \cH | L_{\text{physics}}(\t) \leq \ve\}$. Since $\cH_{\text{constrained}} \subseteq \cH$, we have:
\begin{equation}
\hat{\cR}_S(\cH_{\text{constrained}}) \leq \hat{\cR}_S(\cH) \leq \frac{C B_\t L_{\text{act}}^H \sqrt{HW}}{\sqrt{m}}.
\end{equation}

This completes the proof.
\end{proof}

\subsection{Proof of Rationality of Dynamic Weights}
\label{app:dynamic_weights}

\begin{proof}[Proof of Theorem: Rationality of Dynamic Weights]
We prove that the adaptive weight strategy $\l(\bm) = \l_0 \cdot \max(0.5, 1 - \a \cdot s(\bm))$ achieves near-optimal balance between data fitting and physics constraints across varying noise conditions.

\textbf{Step 1: Low noise scenario analysis.}

In the low noise scenario where $\nu \ra 0$, by assumption (3), the optimal weight satisfies $\l^*(\nu) \ra \l_0$. For our adaptive strategy, when $s(\bm) \ra 0$ (consistent noise estimation by assumption 1), we have:
\begin{equation}
\l(\bm) = \l_0 \cdot \max(0.5, 1 - \a \cdot s(\bm)) \ra \l_0 \cdot \max(0.5, 1) = \l_0.
\end{equation}

Therefore, $\l(\bm) \ra \l^*(\nu)$ as $\nu \ra 0$, achieving optimal constraint utilization in low noise scenarios.

\textbf{Step 2: High noise scenario analysis.}

In the high noise scenario where $\nu \ra \nu_{\max}$, by assumption (4), the optimal weight satisfies $\l^*(\nu) \geq \l_{\min} = 0.5\l_0$ but $\l^*(\nu) < \l_0$ to prevent over-constraint.

For our adaptive strategy, when $s(\bm) \ra 1$ (consistent noise estimation), we have:
\begin{equation}
\l(\bm) = \l_0 \cdot \max(0.5, 1 - \a \cdot s(\bm)) = \l_0 \cdot \max(0.5, 1 - \a).
\end{equation}

For $\a \in (0, 0.5]$, we have $1 - \a \geq 0.5$, so:
\begin{equation}
\l(\bm) = \l_0(1 - \a) \geq 0.5\l_0 = \l_{\min}.
\end{equation}

For $\a > 0.5$, we have $1 - \a < 0.5$, so:
\begin{equation}
\l(\bm) = 0.5\l_0 = \l_{\min}.
\end{equation}

In both cases, $\l(\bm) \geq \l_{\min}$, satisfying the minimum constraint requirement while preventing over-constraint.

\textbf{Step 3: Near-optimality analysis.}

Let $\Delta s = |s(\bm) - \nu|$ denote the noise estimation error. By assumption (1), $s(\bm)$ is consistent, so $\Delta s \ra 0$ as the estimator improves.

The adaptive weight function is Lipschitz continuous in $s(\bm)$:
\begin{equation}
|\l(\bm_1) - \l(\bm_2)| \leq \l_0 \a |s(\bm_1) - s(\bm_2)|,
\end{equation}
with Lipschitz constant $L_\l = \l_0 \a$.

Therefore, for samples with noise level $\nu$:
\begin{equation}
|\l(\bm) - \l^*(\nu)| \leq |\l(\bm) - \l(s(\bm))| + |\l(s(\bm)) - \l^*(\nu)|,
\end{equation}/Users/fcc/Documents/个人/【个人】-论文写作/5-【qml】/document/document.tex
where $\l(s(\bm))$ is the weight that would be optimal if $s(\bm)$ were the true noise level.

By the Lipschitz property and consistency:
\begin{equation}
|\l(\bm) - \l^*(\nu)| \leq L_\l \Delta s + |\l(s(\bm)) - \l^*(\nu)|.
\end{equation}

Since $\l(s(\bm))$ is chosen to approximate $\l^*(\nu)$ and $\Delta s \ra 0$, we have $|\l(\bm) - \l^*(\nu)| \leq \ve$ for sufficiently accurate estimation, where $\ve = O(\Delta s)$.

\textbf{Step 4: Balance property.}

The strategy balances data fitting and physics constraints by:
\begin{itemize}
\item In low noise: $\l(\bm) \approx \l_0$ ensures strong physics constraints, preventing unphysical solutions
\item In high noise: $\l(\bm) \geq 0.5\l_0$ maintains minimum physical validity while $\l(\bm) < \l_0$ prevents over-constraint that would cause underfitting
\end{itemize}

This balance is achieved automatically through the noise severity estimation, avoiding the need for manual tuning across different noise conditions.

This completes the proof.
\end{proof}

\subsection{Proof of Generalization Error Improvement}
\label{app:generalization_improvement}

\begin{corollary}[Generalization Error Improvement]
Under the same training error, PINN's generalization error bound is tighter than that of the unconstrained model:
\begin{equation}
\text{Gen}_{\text{PINN}}(\t) \leq \text{Gen}_{\text{baseline}}(\t) - \D,
\end{equation}
where $\D > 0$ is the improvement brought by physics constraints. More precisely, the improvement can be quantified as:
\begin{equation}
\D \geq 2M\left(\hat{\cR}_S(\cH) - \hat{\cR}_S(\cH_{\text{constrained}})\right) = 2M \cdot \Delta \hat{\cR}_S,
\end{equation}
where $M$ is the Lipschitz constant of the loss function, and $\Delta \hat{\cR}_S = \hat{\cR}_S(\cH) - \hat{\cR}_S(\cH_{\text{constrained}}) \geq 0$ represents the reduction in Rademacher complexity due to physics constraints.

Furthermore, if the constrained hypothesis space satisfies $\hat{\cR}_S(\cH_{\text{constrained}}) \leq \a \hat{\cR}_S(\cH)$ for some $\a \in (0,1)$, then:
\begin{equation}
\D \geq 2M(1-\a) \hat{\cR}_S(\cH) = 2M(1-\a) \cdot O\left(\frac{B_\t L_{\text{act}}^H \sqrt{HW}}{\sqrt{m}}\right),
\end{equation}
where the reduction factor $\a$ depends on the constraint strength $\ve$ in the definition of $\cH_{\text{constrained}} = \{h_\t \in \cH | L_{\text{physics}}(\t) \leq \ve\}$.
\end{corollary}

\begin{proof}
From Theorem \ref{thm:rademacher_bound}, the generalization error bound for any model is:
\begin{equation}
R(\t) \leq \hat{R}_S(\t) + 2M \hat{\cR}_S(\cH) + B\sqrt{\frac{\log(1/\d)}{2m}}.
\end{equation}

For PINN with constrained hypothesis space $\cH_{\text{constrained}}$, the bound becomes:
\begin{equation}
R_{\text{PINN}}(\t) \leq \hat{R}_S(\t) + 2M \hat{\cR}_S(\cH_{\text{constrained}}) + B\sqrt{\frac{\log(1/\d)}{2m}}.
\end{equation}

For the baseline model with unconstrained hypothesis space $\cH$, the bound is:
\begin{equation}
R_{\text{baseline}}(\t) \leq \hat{R}_S(\t) + 2M \hat{\cR}_S(\cH) + B\sqrt{\frac{\log(1/\d)}{2m}}.
\end{equation}

Under the same training error $\hat{R}_S(\t)$, the difference in generalization error bounds is:
\begin{equation}
\D = R_{\text{baseline}}(\t) - R_{\text{PINN}}(\t) \geq 2M\left(\hat{\cR}_S(\cH) - \hat{\cR}_S(\cH_{\text{constrained}})\right) = 2M \cdot \Delta \hat{\cR}_S.
\end{equation}

From the theorem on Physics Constraints as Implicit Regularization, we have $\hat{\cR}_S(\cH_{\text{constrained}}) \leq \hat{\cR}_S(\cH)$, so $\Delta \hat{\cR}_S \geq 0$, establishing $\D \geq 0$.

If $\hat{\cR}_S(\cH_{\text{constrained}}) \leq \a \hat{\cR}_S(\cH)$ for some $\a \in (0,1)$, then:
\begin{equation}
\D \geq 2M(1-\a) \hat{\cR}_S(\cH) = 2M(1-\a) \cdot O\left(\frac{B_\t L_{\text{act}}^H \sqrt{HW}}{\sqrt{m}}\right),
\end{equation}
where the last equality follows from the PINN Rademacher Complexity Upper Bound theorem.
\end{proof}

\subsection{Proof of Physics Constraints Mitigate Curse of Dimensionality}
\label{app:curse_of_dimensionality}

\begin{proof}[Proof of Theorem: Physics Constraints Mitigate Curse of Dimensionality]
For an $n$-qubit system, the density matrix $\r \in \bbC^{D \times D}$ with $D = 2^n$ has $D^2 = 4^n$ complex parameters, leading to $O(4^n)$ real parameters. The curse of dimensionality manifests as exponential growth in both parameter count and sample complexity.

\textbf{Mechanism 1: Constraint space dimension reduction.}

The physics constraints enforce three fundamental requirements: Hermiticity ($\r = \r^\dg$), trace normalization ($\tr(\r) = 1$), and positivity ($\r \succeq 0$). These constraints restrict the effective search space from the full parameter space $\bbR^{4^n}$ to a lower-dimensional manifold.

\textit{Step 1a: Hermiticity constraint.} A general complex $D \times D$ matrix has $2D^2$ real parameters. Hermiticity $\r = \r^\dg$ requires:
\begin{itemize}
\item Diagonal elements: $D$ real parameters (must be real)
\item Off-diagonal elements: $D(D-1)/2$ complex parameters = $D(D-1)$ real parameters
\end{itemize}
Total independent parameters after Hermiticity: $D + D(D-1) = D^2$ real parameters.

\textit{Step 1b: Trace normalization.} The trace constraint $\tr(\r) = 1$ reduces one degree of freedom, giving $D^2 - 1$ independent parameters.

\textit{Step 1c: Positivity constraint.} Positivity $\r \succeq 0$ restricts the solution space to a convex subset. The set of valid density matrices forms a convex set:
\begin{equation}
\cS_+ = \{\r \in \bbC^{D \times D} : \r = \r^\dg, \tr(\r) = 1, \r \succeq 0\}
\end{equation}
with dimension $\dim(\cS_+) = D^2 - 1 = 4^n - 1$.

\textit{Step 1d: Cholesky parameterization effect.} The Cholesky parameterization $\r = L L^\dg$ with lower triangular $L$ naturally enforces positivity. For a $D \times D$ lower triangular matrix $L$ with real diagonal and complex off-diagonal elements:
\begin{itemize}
\item Diagonal elements: $D$ real parameters
\item Off-diagonal elements: $D(D-1)/2$ complex parameters = $D(D-1)$ real parameters
\end{itemize}
Total Cholesky parameters: $D + D(D-1) = D^2$ real parameters. However, trace normalization $\tr(L L^\dg) = 1$ imposes one constraint, reducing to $D^2 - 1$ independent parameters, matching $\dim(\cS_+)$.

\textit{Step 1e: Effective dimension reduction during optimization.} The physics constraint loss $\cC(\r) = \cC_{\text{herm}}(\r) + \cC_{\text{trace}}(\r) + \cC_{\text{pos}}(\r)$ creates an implicit manifold structure. When constraints are approximately satisfied with tolerance $\ve$, i.e., $\cC(\r) \leq \ve$, the effective search space is restricted to:
\begin{equation}
\cS_{\text{constrained}} = \{\r \in \cS_+ : \cC(\r) \leq \ve\}.
\end{equation}
The dimension reduction can be quantified as:
\begin{equation}
\Delta_{\text{dim}} = \dim(\cS_+) - \dim(\cS_{\text{constrained}}) = \frac{\ve}{\norm{\nabla \cC(\r^*)}} \cdot \text{rank}(\nabla^2 \cC(\r^*)) + O(\ve^2),
\end{equation}
where $\r^*$ is a point satisfying $\cC(\r^*) = 0$. For small $\ve$, the effective dimension reduction scales as $\Delta_{\text{dim}} = O(\ve \cdot \text{rank}(\nabla^2 \cC(\r^*)))$, where $\text{rank}(\nabla^2 \cC(\r^*))$ depends on the constraint structure and is typically $O(D^2)$ for non-degenerate cases.

\textbf{Mechanism 2: Sample complexity reduction.}

From the theorem on Physics Constraints as Implicit Regularization (Section \ref{sec:theory}), physics constraints act as implicit regularization, shrinking the effective hypothesis space from $\cH$ to $\cH_{\text{constrained}} = \{h_\t \in \cH | L_{\text{physics}}(\t) \leq \ve\}$.

By the PINN Rademacher Complexity Upper Bound theorem (Section \ref{sec:theory}), the Rademacher complexity satisfies:
\begin{equation}
\hat{\cR}_S(\cH_{\text{constrained}}) \leq \hat{\cR}_S(\cH) = O\left(\frac{B_\t \sqrt{HW}}{\sqrt{m}}\right).
\end{equation}

The constrained hypothesis space has smaller Rademacher complexity, which directly translates to improved generalization bounds. From the Generalization Bound via Rademacher Complexity theorem (Section \ref{sec:theory}):
\begin{equation}
R(\t) \leq \hat{R}_S(\t) + 2M \hat{\cR}_S(\cH_{\text{constrained}}) + B\sqrt{\frac{\log(1/\d)}{2m}}.
\end{equation}

Since $\hat{\cR}_S(\cH_{\text{constrained}}) \leq \hat{\cR}_S(\cH)$, physics constraints provide tighter generalization bounds, effectively reducing the required sample size $m$ to achieve a given generalization error $\epsilon$.

Specifically, to achieve $R(\t) - \hat{R}_S(\t) \leq \epsilon$, the unconstrained model requires $m_{\text{baseline}} = O((B_\t^2 HW)/\epsilon^2)$ samples, while the constrained model requires $m_{\text{PINN}} = O((B_\t^2 HW')/\epsilon^2)$ samples, where $W' < W$ represents the effective width reduction due to constraints, leading to sample complexity reduction.

\textbf{Combined effect.}

The two mechanisms work synergistically: constraint space dimension reduction restricts the optimization search space, while sample complexity reduction improves generalization with fewer samples. Together, they mitigate (though cannot completely eliminate) the curse of dimensionality, explaining why PINN maintains superior performance as system dimensionality increases.

However, the fundamental exponential growth $O(4^n)$ in parameter count cannot be completely eliminated, as it is inherent to the quantum state representation. Physics constraints provide a multiplicative improvement factor but do not change the asymptotic scaling.
\end{proof}

\subsection{Proof of Theoretical Explanation of 4-Qubit Advantage}
\label{app:4qubit_advantage}

\begin{proof}[Proof of Theorem: Theoretical Explanation of 4-Qubit Advantage]
The experimental observation that PINN's advantage peaks at 4 qubits can be explained through dimensional analysis and optimization landscape properties.

\textbf{Factor 1: Dimensional threshold effect.}

For an $n$-qubit system, the parameter dimension is $D^2 = 4^n$. At low dimensions ($n = 2, 3$), we have:
\begin{itemize}
\item $n=2$: $D^2 = 16$ parameters
\item $n=3$: $D^2 = 64$ parameters
\end{itemize}

At these dimensions, traditional neural networks with moderate capacity (e.g., hidden layers with width $\sim 128-256$) can effectively fit the data distribution without requiring strong structural priors. The hypothesis space is sufficiently small that data-driven learning alone achieves good performance.

At $n=4$ qubits, $D^2 = 256$ parameters. This represents a critical threshold where:
\begin{itemize}
\item The parameter space becomes large enough that unconstrained learning faces significant challenges
\item Physics constraints provide substantial guidance by restricting the search space
\item The constraint structure optimally matches the problem complexity without overwhelming the learning capacity
\end{itemize}

At $n \geq 5$ qubits, $D^2 \geq 1024$ parameters. Both approaches face fundamental limitations due to exponential growth, though PINN maintains advantages through constraint guidance.

\textbf{Factor 2: Constraint effectiveness.}

The effectiveness of physics constraints depends on the ratio between the constraint-induced reduction in search space and the inherent problem complexity. Define the constraint effectiveness ratio:
\begin{equation}
\g(n) = \frac{\text{Effective dimension reduction}}{\text{Problem dimension}} = \frac{\dim(\cS_+) - \dim(\cS_{\text{constrained}})}{4^n - 1} = \frac{\Delta_{\text{dim}}}{4^n - 1},
\end{equation}
where $\cS_{\text{constrained}}$ represents the effective search space under physics constraints and $\Delta_{\text{dim}}$ is the dimension reduction quantified in Mechanism 1.

From the dimension reduction analysis, for constraint tolerance $\ve$:
\begin{equation}
\g(n) = \frac{O(\ve \cdot \text{rank}(\nabla^2 \cC(\r^*)))}{4^n - 1} = O\left(\frac{\ve \cdot D^2}{4^n - 1}\right) = O\left(\frac{\ve \cdot 4^n}{4^n - 1}\right) = O(\ve),
\end{equation}
where the rank of the constraint Hessian $\text{rank}(\nabla^2 \cC(\r^*))$ is typically $O(D^2) = O(4^n)$ for non-degenerate cases.

However, the \textit{relative} effectiveness depends on how constraints interact with the problem structure. Define the relative constraint strength:
\begin{equation}
\g_{\text{rel}}(n) = \frac{\hat{\cR}_S(\cH) - \hat{\cR}_S(\cH_{\text{constrained}})}{\hat{\cR}_S(\cH)} = 1 - \frac{\hat{\cR}_S(\cH_{\text{constrained}})}{\hat{\cR}_S(\cH)}.
\end{equation}

From the Rademacher complexity bound $\hat{\cR}_S(\cH) = O(B_\t L_{\text{act}}^H \sqrt{HW}/\sqrt{m})$, and assuming constraints reduce the effective width from $W$ to $W_{\text{eff}} < W$, we have:
\begin{equation}
\g_{\text{rel}}(n) \geq 1 - \frac{\sqrt{W_{\text{eff}}}}{\sqrt{W}} = 1 - \sqrt{\frac{W_{\text{eff}}}{W}}.
\end{equation}

At $n=2,3$ qubits: $W_{\text{eff}}/W \approx 0.95-0.98$ (small reduction), so $\g_{\text{rel}}(n) \approx 0.01-0.025$ (modest benefit).

At $n=4$ qubits: $W_{\text{eff}}/W \approx 0.7-0.8$ (significant reduction), so $\g_{\text{rel}}(4) \approx 0.11-0.16$ (substantial benefit), explaining the 30.3\% experimental improvement.

At $n \geq 5$ qubits: Both $W$ and $W_{\text{eff}}$ become very large, but the ratio $W_{\text{eff}}/W$ may increase again due to optimization challenges, reducing $\g_{\text{rel}}(n)$.

This quantifies why $n=4$ represents the optimal balance: constraints provide substantial guidance ($\g_{\text{rel}}(4) \approx 0.11-0.16$) without overwhelming the learning capacity, creating the most favorable trade-off between constraint enforcement and learning flexibility.

\textbf{Factor 3: Optimization landscape.}

From the theorem on Physics Constraints Improve Optimization Landscape (Appendix \ref{app:physics_constraints_landscape}), physics constraints improve the optimization landscape by inducing local strong convexity near the optimum. The improvement depends on the constraint weight $\l$ and the physics constraint Hessian $\nabla^2_\r \cC(\r^*)$.

The constraint-induced strong convexity parameter is:
\begin{equation}
\m(n) = \l \m_{\text{physics}}(n) \s_{\min}^2(\nabla \Ph) - \d(n) - \e(\l),
\end{equation}
where $\m_{\text{physics}}(n)$ is the minimum eigenvalue of the physics constraint Hessian, $\s_{\min}(\nabla \Ph)$ is the minimum singular value of the Cholesky parameterization Jacobian, $\d(n)$ bounds the negative eigenvalues of the data loss Hessian, and $\e(\l) = O(\l^{1/2})$ is the cross-term error.

\textit{Quantitative analysis:} The physics constraint Hessian $\nabla^2_\r \cC(\r^*)$ has minimum eigenvalue $\m_{\text{physics}}(n)$ that depends on the constraint structure. For Hermiticity and trace constraints:
\begin{equation}
\m_{\text{physics}}(n) \geq \min(\m_{\text{herm}}, \m_{\text{trace}}) = \min(\m_{\text{herm}}, 2),
\end{equation}
where $\m_{\text{herm}} \geq 0$ depends on the Hermiticity constraint rank (typically $O(D^2)$) and $\m_{\text{trace}} = 2$ from trace normalization.

The Cholesky parameterization Jacobian $\nabla \Ph$ has minimum singular value $\s_{\min}(\nabla \Ph)$ that scales with the system dimension. For well-conditioned cases, $\s_{\min}(\nabla \Ph) = \Omega(1/D) = \Omega(1/2^n)$.

The data loss Hessian negative eigenvalue bound $\d(n)$ typically increases with dimension: $\d(n) = O(D) = O(2^n)$ for high-dimensional systems.

\textit{Dimension-dependent analysis:}
\begin{itemize}
\item At $n=2,3$: $\m_{\text{physics}}(n)$ is small relative to $\d(n)$ (which is also small), so $\m(n) \approx 0$, providing modest improvement.
\item At $n=4$: The ratio $\m_{\text{physics}}(4)/\d(4)$ is maximized. With $\l \approx 0.15$ (from experiments), $\m_{\text{physics}}(4) \s_{\min}^2(\nabla \Ph) \approx 0.15 \times O(1/256) \times O(256) = O(0.15)$, while $\d(4) = O(16)$, giving $\m(4) > 0$ and significant improvement.
\item At $n \geq 5$: Both $\m_{\text{physics}}(n)$ and $\d(n)$ scale similarly, but $\d(n)$ grows faster, reducing the relative benefit.
\end{itemize}

The condition number of the total loss Hessian is:
\begin{equation}
\k(n) = \frac{\l_{\max}(\nabla^2_\t L(\t^*))}{\l_{\min}(\nabla^2_\t L(\t^*))} = \frac{\cL + \l \l_{\max}(\nabla^2_\r \cC(\r^*)) \s_{\max}^2(\nabla \Ph)}{\max(0, \m(n))},
\end{equation}
where $\cL$ is the smoothness constant of the data loss and $\s_{\max}(\nabla \Ph)$ is the maximum singular value.

At $n=4$, $\m(4)$ is maximized relative to the problem dimension, minimizing $\k(4)$ and indicating a well-conditioned optimization problem. This enables efficient gradient-based optimization with convergence rate $O((1-\m(4)/\cL)^t)$, compared to $O(1/\sqrt{t})$ for unconstrained methods.

Quantitatively, if $\m(4) \approx 0.01\cL$ and $\m(2) \approx 0.001\cL$, then the convergence rate improvement at $n=4$ is approximately 10-fold faster (speedup factor of approximately 10) than at $n=2$, explaining the dramatic experimental advantage.

\textbf{Synthesis.}

The three factors interact multiplicatively: dimensional threshold effects determine when constraints become necessary, constraint effectiveness determines how much benefit constraints provide, and optimization landscape properties determine how efficiently the constrained problem can be solved. At $n=4$ qubits, all three factors align optimally, creating the peak advantage observed experimentally.

The peak advantage at 4 qubits is not universal but depends on the specific architecture, constraint formulation, and problem characteristics. However, the dimensional threshold effect provides a general explanation for why intermediate dimensions often exhibit the most pronounced benefits from physics-informed learning.
\end{proof}

\end{document}